\documentclass[sigplan, screen]{acmart}

\setcopyright{none}

\usepackage{microtype}
\usepackage{booktabs}
\usepackage{subcaption}
\usepackage[noend]{algpseudocode}
\usepackage[plain,figure]{algorithm}
\usepackage{amsmath}

\usepackage{stmaryrd}
\usepackage{url}
\usepackage{listings}
\usepackage{mdwlist}
\usepackage{pifont}
\usepackage{graphicx}
\usepackage{wrapfig}
\usepackage{fancyvrb}
\usepackage{paralist}
\usepackage{caption}
\usepackage{comment}
\usepackage{bbm}
\usepackage{mathrsfs}
\usepackage{xcolor}
\usepackage{tikz}
\usetikzlibrary{patterns}
\usepackage{pgfplots}
\usepackage{esvect}
\usepackage{xspace}
\usepackage{array, multirow}
\usepackage{rotating, makecell}
\usepackage{enumitem}
\usepackage{makecell}
\usepackage[toc,page]{appendix}
\usepackage{todonotes}
\usepackage{multicol}
\usepackage{lipsum}
\usepackage{stfloats}
\usepackage{balance}

\begin{document}

\title[Web Question Answering with Neurosymbolic Program Synthesis]{Web Question Answering \\ with Neurosymbolic Program Synthesis}

\author{Qiaochu Chen}
\affiliation{
  \institution{University of Texas at Austin} 
  \city{Austin}
  \state{Texas}
  \country{USA}
}
\email{qchen@cs.utexas.edu}

\author{Aaron Lamoreaux}
\affiliation{
  \institution{University of Texas at Austin}
  \city{Austin}
  \state{Texas}
  \country{USA} 
}
\email{lamoreauxaj@gmail.com}

\author{Xinyu Wang}
\affiliation{
  \institution{University of Michigan}
  \city{Ann Arbor}
  \state{Michigan}
  \country{USA}
}
\email{xwangsd@umich.edu}

\author{Greg Durrett}
\affiliation{
  \institution{University of Texas at Austin} 
  \city{Austin}
  \state{Texas}
  \country{USA}
}
\email{gdurrett@cs.utexas.edu} 

\author{Osbert Bastani}
\affiliation{
  \institution{University of Pennsylvania}
  \city{Philadelphia}
  \state{Pennsylvania}
  \country{USA}
}
\email{obastani@seas.upenn.edu}

\author{Isil Dillig}
\affiliation{
  \institution{University of Texas at Austin}
  \city{Austin}
  \state{Texas}
  \country{USA}
}
\email{isil@cs.utexas.edu}

\newcommand{\assign}{\gets}
\newcommand{\res}{R}
\newcommand{\bestf}{\emph{opt}}
\newcommand{\guards}{G}
\newcommand{\extractors}{E}
\newcommand{\dsl}{{\sc WebQA}\xspace}

\newcommand{\question}{Q}
\newcommand{\keywordss}{K}
\newcommand{\webpage}{W}
\newcommand{\guard}{\psi}
\newcommand{\prog}{p}
\newcommand{\stringExt}{e}
\newcommand{\nodeExt}{\nu}
\newcommand{\neural}{\phi}
\newcommand{\nfilter}{\varphi}
\newcommand{\fonescore}{s}
\newcommand{\semantics}[1]{\llbracket{#1}\rrbracket}                    
\newcommand{\jc}[1]{\textcolor{red}{\textbf{JC:} #1}}
\newcommand{\greg}[1]{\textcolor{red}{\textbf{GD:} #1}}
\newcommand{\xinyu}[1]{\textcolor{red}{\textbf{XW:} #1}}
\newcommand{\osbert}[1]{\textcolor{red}{\textbf{OB:} #1}}
\newcommand{\aaron}[1]{\textcolor{red}{\textbf{AL:} #1}}
\newcommand{\red}[1]{\textcolor{red}{#1}}

\newcommand{\toolname}{{\sc WebQA}\xspace}
\newcommand{\toolnameD}{\toolname-{\sc NoDecomp}\xspace}
\newcommand{\toolnameP}{\toolname-{\sc NoPrune}\xspace}
\newcommand{\tableT}{{\tt Table}\xspace}
\newcommand{\listT}{{\tt List}\xspace}
\newcommand{\stringT}{{\tt String}\xspace}
\newcommand{\setT}{{\tt Set}\xspace}

\newcommand{\figref}[1]{Figure~\ref{#1}}

\algnewcommand\Input{\textbf{input: }}
\algnewcommand\Output{\textbf{output: }}

\newcommand{\partition}{P}
\newcommand{\examples}{\mathcal{E}}
\newcommand{\examplee}{e}
\newcommand{\worklist}{\mathcal{W}}

\newcommand\numberthis{\addtocounter{equation}{1}\tag{\theequation}}

\newcommand{\irule}[2]{\mkern-2mu\displaystyle\frac{#1}{\vphantom{,}#2}\mkern-2mu}

\newcommand{\inputexamples}{\mathcal{I}}
\newcommand{\outputexamples}{\mathcal{O}}
\newcommand{\inputexamplee}{i}
\newcommand{\outputexamplee}{o}
\begin{abstract}
In this paper, we propose a new technique based on program synthesis for extracting information from webpages. Given a natural language query and a few labeled webpages, our method synthesizes a program that can be used to extract similar types of information from other unlabeled webpages. To handle websites with diverse structure, our approach employs a neurosymbolic DSL that incorporates both neural NLP models as well as standard language constructs for tree navigation and string manipulation. We also propose an optimal synthesis algorithm that generates all DSL programs that achieve optimal $F_1$ score on the training examples. Our synthesis technique is compositional, prunes the search space by exploiting a  monotonicity property of the DSL, and uses transductive learning to select programs with good generalization power. We have implemented these ideas in a new tool called \toolname and evaluate  it on 25 different tasks across multiple domains. Our experiments show that \toolname significantly outperforms existing tools such as state-of-the-art question answering models and wrapper induction systems. 
\end{abstract}



\maketitle

\section{Introduction}\label{sec:intro}

As the amount of information available on the web proliferates, there is a growing need for tools that can extract relevant information from  websites. Due to the importance of this problem, there has been a flurry of research activity on \emph{information extraction}~\cite{liang2014, lockard20zeroshotceres} and \emph{wrapper induction}~\cite{muslea1999hierarchical, hsu1998generating, chang2001iepad, crescenzi2001roadrunner, anton2005xpath, raza2020web, flashextract, gulhane2011web}. In particular, most recent research from the natural language processing (NLP) community  focuses on unstructured text documents and employs powerful neural models to automate information extraction and question answering (QA) tasks. On the other hand, most  wrapper induction work 
focuses on semi-structured documents and aims to synthesize programs (e.g., XPath queries) to extract relevant nodes from the DOM tree. While such wrapper induction techniques work well when the target webpages have a shared global schema (e.g., Yelp pages or LinkedIn profiles), they are not as effective on structurally heterogeneous websites such as faculty webpages. On the other hand, ML-based techniques from the NLP community are, in principal, applicable to  heterogeneous websites; however, by treating the entire webpage as unstructured text, they fail to take advantage of the inherent tree structure of HTML documents. 

In this paper, we propose a new information extraction approach ---based on \emph{neurosymbolic program synthesis} --- that combines the relative strengths of wrapper induction techniques for webpages with the flexibility of neural models for unstructured documents. Our approach targets structurally heterogeneous websites with no shared global schema and can be used to automate many different types of information extraction tasks. Similar to prior program synthesis approaches\cite{fidex, flashextract}, our approach can learn useful extractors from a \emph{small} number of labeled webpages.

\begin{figure}
    \includegraphics[scale=0.24, trim=237 500 500 270, clip ]{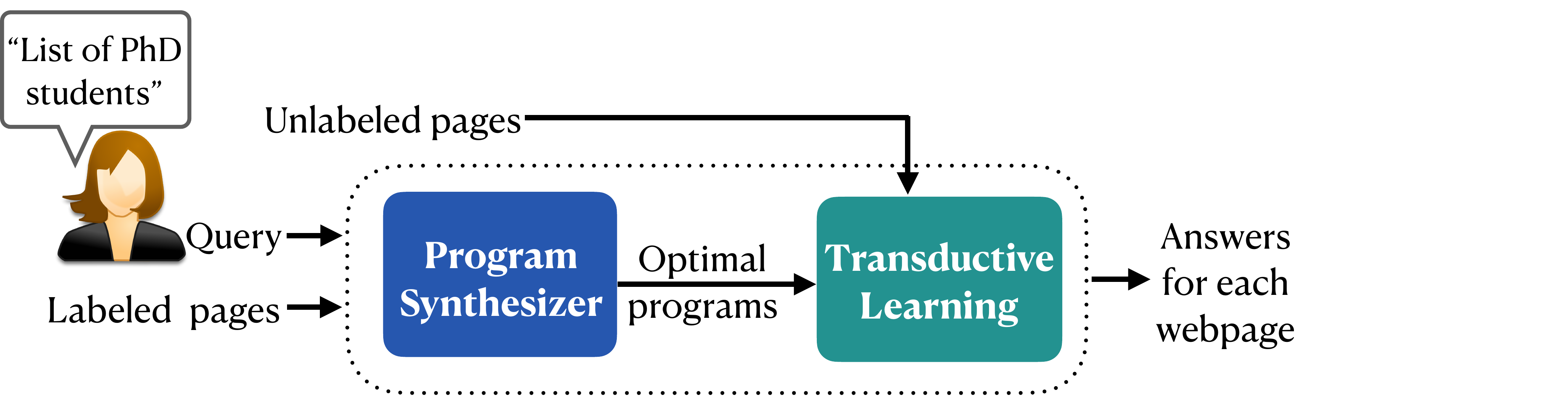}
    \caption{Schematic overview of our approach}
    \label{fig:workflow}
\end{figure}

As illustrated in Figure~\ref{fig:workflow}, our approach takes three inputs, including (1) a natural language query, (2) a small number of labeled webpages, and (3) a much bigger set of unlabeled webpages from which to extract information. For instance, if the task is to extract PhD students from faculty webpages, the input might consist of a question such as ``Who are the PhD students?" as well as keywords like "advisees'' and "PhD students''. In addition, the user would also provide a set of target faculty webpages, together with labels (i.e., names of PhD students) for a few of these. Given this input, the goal of our technique is to generate a program that can be used to extract the desired information from all target webpages.


To solve this challenging problem, we employ a multi-pronged solution that incorporates three key ingredients:
\begin{itemize}[leftmargin=*]
    \item {\bf Neurosymbolic DSL: } To combine the relative strengths of wrapper induction techniques with the flexibility of language models, we design a new neurosymbolic domain-specific language targeted for web question answering. Our DSL combines pre-trained neural modules for natural language processing with standard programming language constructs for string processing and tree traversal.
    \item {\bf Optimal program synthesis:} To utilize this DSL for automated web information extraction, we describe a new program synthesis technique for finding DSL programs that best fit the labeled webpages. However, since it is often impossible to find programs that \emph{exactly} fit the provided labels, we instead search for programs that optimize $F_1$ score\footnote{{ $F_1$ score is computed as  $2 \cdot \frac{\texttt{precision} \cdot \texttt{recall}}{\texttt{precision} + \texttt{recall}}$. It is a common evaluation metric in information extraction. }}. Our proposed optimal synthesis method is compositional and leverages   a  monotonicity property of the DSL to aggressively prune parts of the search space that are guaranteed \emph{not} to contain an optimal program.  
    \item {\bf Transductive program selection:} During synthesis, there are often \emph{many} (e.g., hundreds of) DSL programs with optimal $F_1$ score on the labeled data. However, not all of these candidate programs perform well on  test data, and standard heuristics (e.g., based on program size) are not effective at distinguishing between these programs. We address this challenge using transductive learning: it generates \emph{soft labels}  for the test data based on all candidate programs and then chooses the ``consensus'' program whose output most closely matches the soft labels.
\end{itemize}

We have implemented our proposed approach in a tool called \toolname{}
and evaluate it across  several different tasks  and many webpages. Our evaluation demonstrates that \toolname{} yields significantly better results compared to existing baselines,  including both question answering models and wrapper induction systems. We also perform ablation studies to evaluate the relative importance of our proposed techniques and show that all of these ideas are important for making this approach practical.

In summary, this paper makes the following  contributions:

\begin{itemize}[leftmargin=*]
    \item We propose a new technique for web question answering that is based on optimal neurosymbolic program synthesis.
    \item We present a DSL for web information extraction that combines pre-trained NLP models with traditional language constructs for string manipulation and tree traversal.
    \item We describe a compositional program synthesis technique for finding all programs that achieve optimal $F_1$ score on the labeled webpages. Our synthesis algorithm prunes the search space by exploiting a monotonicity property of the DSL with respect to recall.
    \item We present a transductive learning technique for choosing a good program for labeling the target webpages.
    \item We implement our approach in a tool called \toolname{} and evaluate it on 25 different tasks spanning four domains and 160 webpages.
\end{itemize}
\section{Motivating Example}\label{sec:motivating}

In this section, we present a motivating scenario for  \toolname\  and highlight salient features of our approach.

\begin{figure*}[!t]
\hspace{-15pt}
\begin{center}
    \begin{minipage}[t]{0.63\linewidth}
    \includegraphics[width=0.48\linewidth,trim=180 300 200 20, clip]{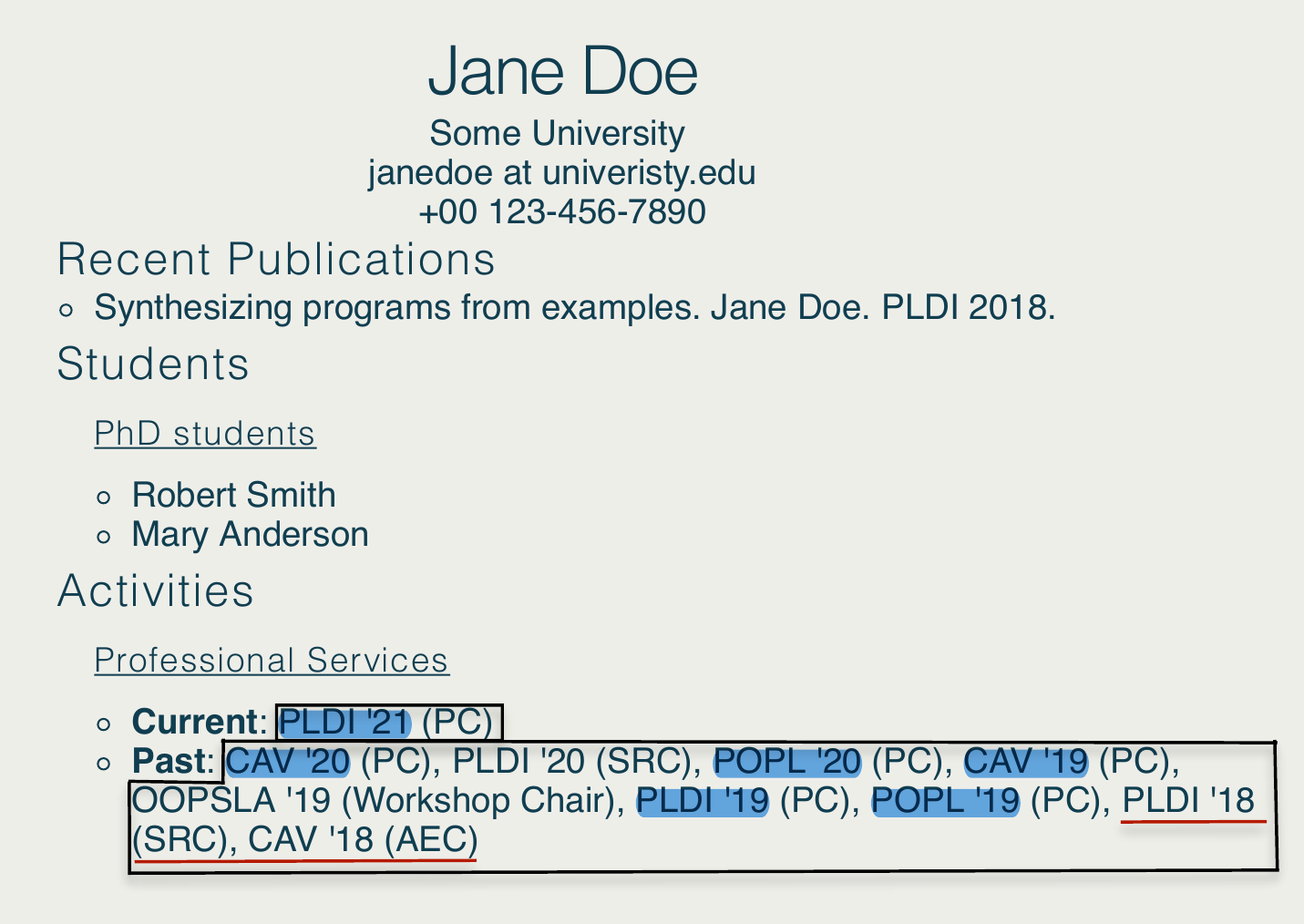}
    \includegraphics[width=0.5\linewidth,trim=170 290 170 0, clip]{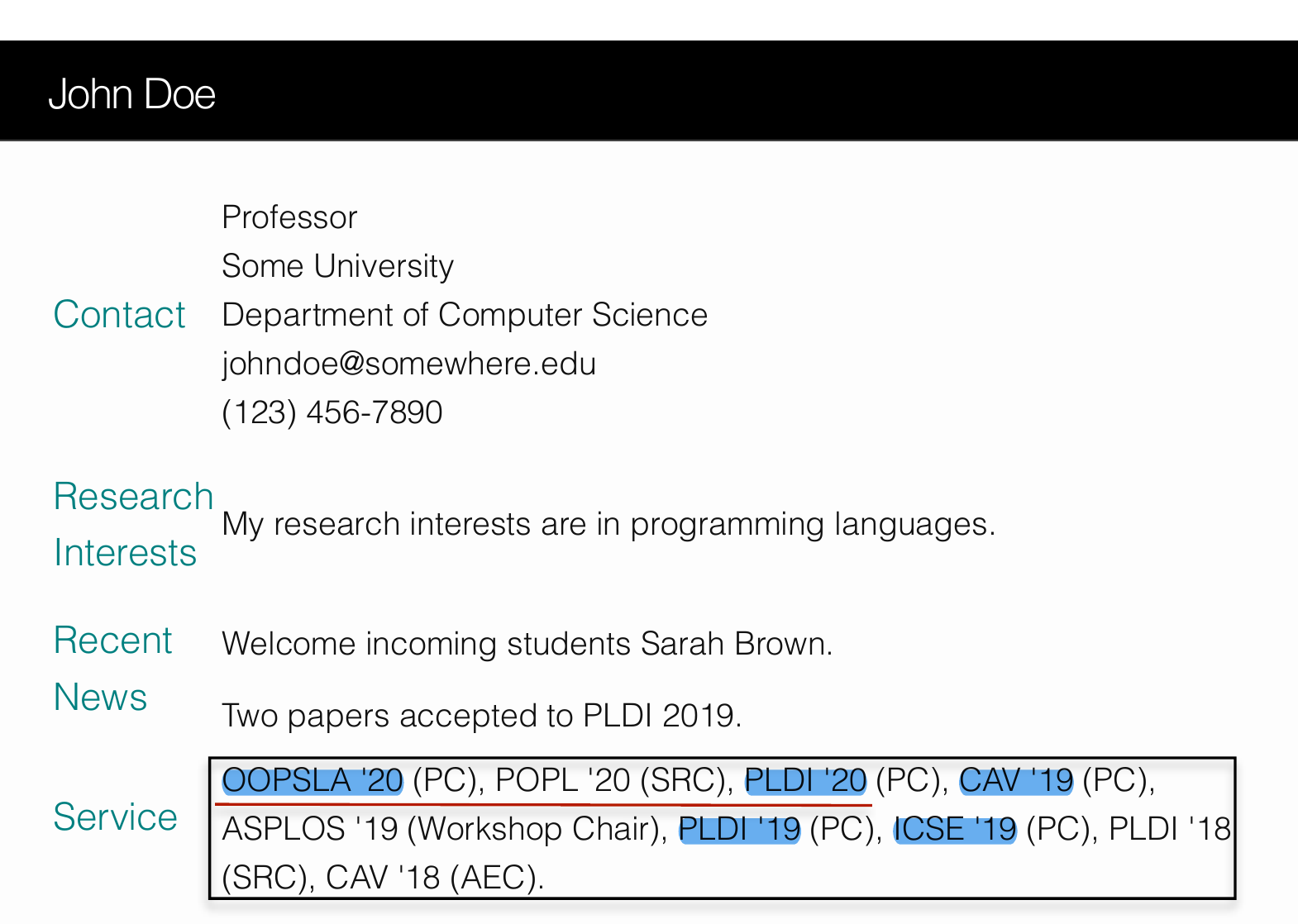}
    \caption{Sample faculty websites with their program committee information. Correct answers are in blue; the output of a QA model is underlined in red.}
    \label{fig:webpages}
    \end{minipage}
    \hspace{20pt}
    \begin{minipage}[t]{0.30\linewidth}
    \includegraphics[width=1.0\linewidth, trim=200 310 200 0, clip]{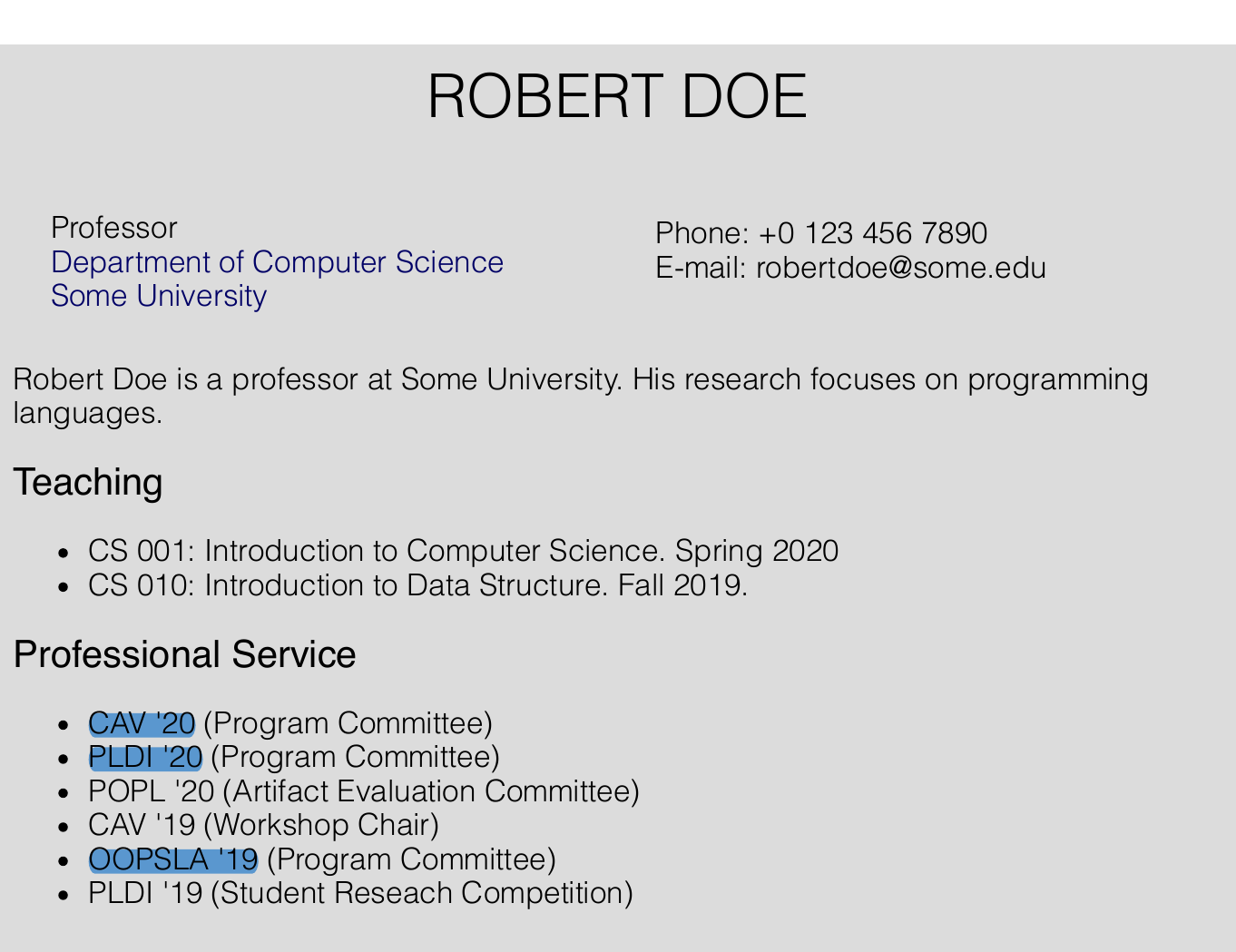}
    \caption{The synthesized program also works on this website; highlighted text are the extracted output.}
    \label{fig:webpages2}
    \end{minipage}
\end{center}

\end{figure*}



\noindent
\paragraph{Usage scenario.} Suppose that the PC chair for a conference needs to form a program committee, and she has access to the websites of many researchers. To help her form a good committee, she wants to extract program committees that each researcher has served on (which is often available on their websites). Since there are too many websites, extracting this information manually is too laborious. 
Our proposed system, \toolname, is useful in scenarios like this
that require collecting information from many structurally heterogeneous websites. 




To use \toolname, the user starts by providing a question (e.g., ``Which program committees has this researcher served on?'') and a set of keywords (e.g., ``PC'', ``Program Committee'', ``Service''). Then, given a target set of websites, \toolname\ asks the user to provide labels for a small number of webpages. For instance,  Figure~\ref{fig:webpages} presents two (hypothetical) websites that \toolname may show to the user, with the user-provided labels highlighted in blue. Observe that both of these webpages are semi-structured in the sense that they contain clearly-delineated sections (e.g., Students, Service); however, they differ both in terms of their high-level structure and what information they contain. 


\paragraph{Limitations of existing approaches.} We now use this simple motivating example to illustrate why existing approaches are not effective for this type of tasks. As mentioned in Section~\ref{sec:intro}, there are two classes of techniques, namely \emph{program induction} and \emph{question answering}, that could potentially be useful in this setting. 

Like our approach, program induction techniques aim to extract information from webpages based on a small number of user-provided training examples~\cite{raza2020web,iyer2019synthesis}.  Specifically, given a few  labeled webpages, these techniques learn XPath expressions to locate  relevant nodes in the DOM tree. However, as illustrated in Figure~\ref{fig:webpages}, researcher webpages typically do not have a uniform structure. Furthermore, even for webpages that are structurally somewhat similar, they exhibit minor variations (e.g., different section names, relative ordering of sections etc.) that make it very difficult to learn XPath expressions that   generalize well to unseen websites. In addition, almost all existing techniques in this space focus on extracting relevant nodes in the DOM tree; however, they do not attempt to perform any further text processing within that node. 
As illustrated by both webpages in Figure~\ref{fig:webpages}, extracting the desired information requires further processing at the text level, such as extracting relevant substrings.

An alternative approach for automating this task is to use a state-of-the-art question answering (QA) system that treats the entire webpage as a raw sequence of words.
However, in practice, such approaches perform poorly since they are not designed to leverage the tree structure of the document. 
Furthermore, because they treat  text across different DOM nodes as natural language,  they have difficulty dealing with more structured information like long comma-delineated lists or formatting with parentheticals. 
For instance, for the two webpages from Figure~\ref{fig:webpages} and the question ``Which program committees has this researcher served on?'', a BERT-based QA system~\cite{bert} yields the suboptimal answers underlined in red in Figure~\ref{fig:webpages}. In particular, it either outputs incorrect spans or includes  text that should not be part of the answer (e.g. ``POPL'20 (SRC)'' in the second webpage).

\paragraph{Key idea \#1: Neurosymbolic DSL} Our approach combines the relative strengths of machine learning and program induction techniques by synthesizing programs in a neurosymbolic DSL for web information extraction.  In particular, our proposed DSL incorporates both pre-trained neural models for question answering, keyword matching, and entity extraction with standard programming language constructs for string processing and tree navigation.  The tree navigation constructs allow taking advantage of webpage structure, while making it possible to handle minor variations (e.g., exact section names) using pre-trained neural models. Furthermore, the presence of string processing constructs in the DSL allows our method to extract fine-grained information within individual tree nodes. 

In more detail, a program in our DSL is structured to first locate relevant nodes in the tree representation of a webpage (see Figure~\ref{fig:webtree}) and then perform additional information extraction from each tree node.  For example, the following  code snippet in our DSL can be used to locate the relevant parts of the webpages from Figure~\ref{fig:webpages}:
{
\small
\begin{align*}
         & {\tt \small GetLeaves}({\tt \small GetDescendents}(r, \lambda z. {\tt \small matchKeyword}(z, \keywordss)))
        \numberthis \label{eq:locator1}
\end{align*}}
\normalsize
\begin{figure}[t]
    \centering
    \includegraphics[width=0.43\textwidth, trim=400 200 200 50, clip]{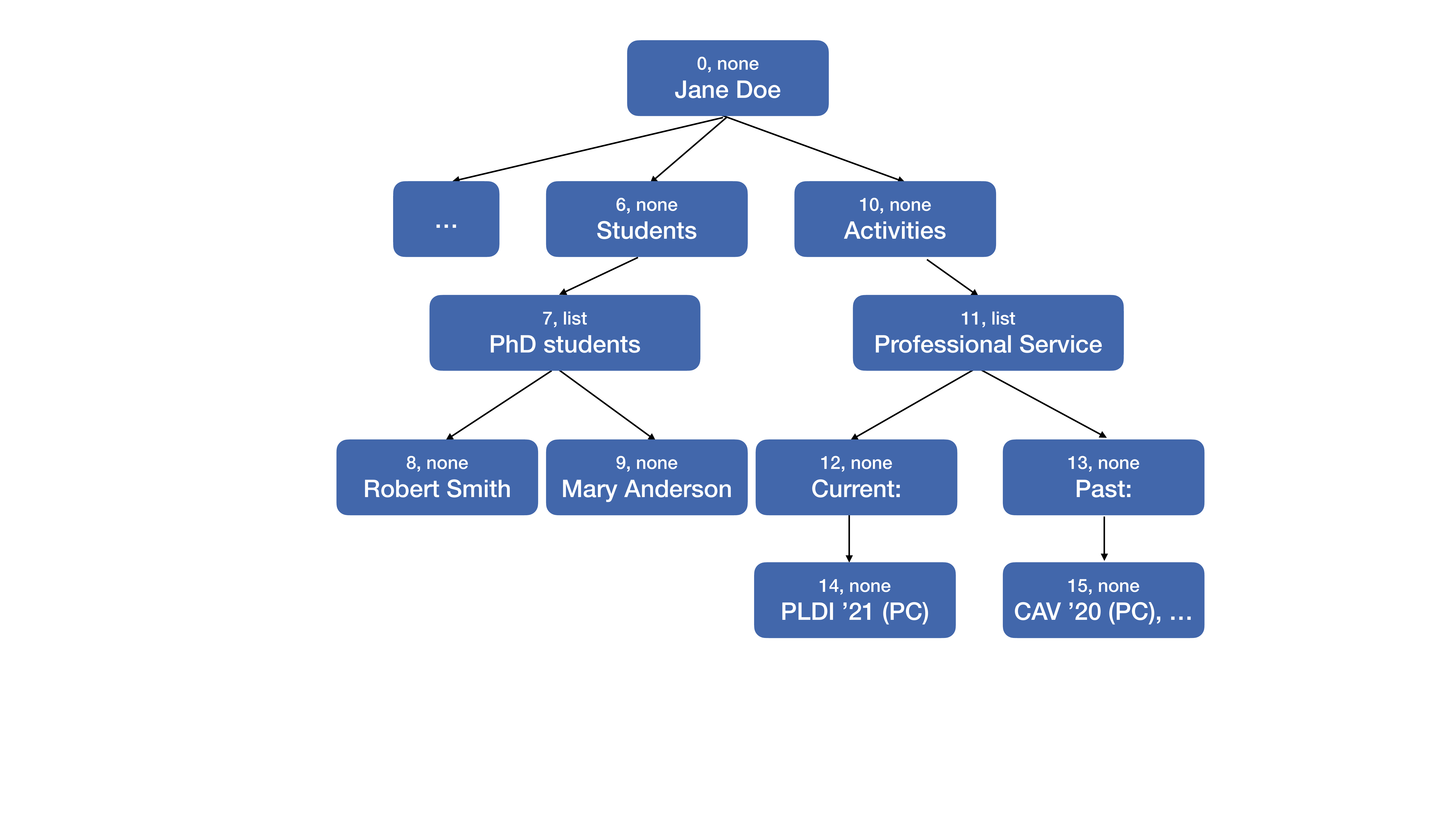}
    \caption{Tree representation of top webpage from Figure~\ref{fig:webpages}. Each node contains the node id, type and its content.}
    \label{fig:webtree}
\end{figure}

Here, $r$ is the root node of the input webpage, and the construct {\tt GetDescendants}($r, \phi$)  returns all tree nodes whose content satisfies predicate $\phi$. In the code snippet above, the predicate $\lambda z. \texttt{matchKeyword}(z, K)$ is implemented by a neural network  that has been pre-trained for keyword matching. Thus, this program first locates all tree nodes whose content matches any of the provided keywords $K$ and returns all of their leaf nodes. For example, given the tree  in Figure~\ref{fig:webtree} representing the top webpage from \figref{fig:webpages}, the {\tt GetDescendants} sub-program will match node $11$, and {\tt GetLeaves}\footnote{Actually, there is no explicit {\tt GetLeaves}(v) construct in our DSL; this is just syntactic sugar for {\tt GetDescendants}(v, $\lambda$ n. {\tt isLeaf}(n)). } will return  nodes $14$ and $15$, which are  leaf nodes of the subtree rooted at node $11$. For the webpages in Figure~\ref{fig:webpages}, this program yields the tree nodes annotated using black boxes. 

Next, given the tree nodes returned by the above code snippet, we can extract the desired information from these nodes using the following code snippet in our DSL: 
\small
\begin{align*}
         &\lambda x. 
         {\tt GetEntity}({\tt Filter}({\tt Split}({\tt ExtractContent}(x),  {\tt COMMA}), \\
         & \qquad \qquad \qquad \qquad \qquad \qquad \lambda z. {\tt matchKeyword}(z, K)), {\tt ORG}) \numberthis \label{eq:extractor1}
\end{align*}
\normalsize
In particular, this code snippet first retrieves the content of tree node $x$ using  {\tt ExtractContent}  and then splits it into a set of (comma-separated) strings using {\tt Split}. Then, it filters those elements that do not match the provided keywords and finally extracts substrings that correspond to an organization entity\footnote{{Note that there is no explicit ${\tt GetEntity}$ on our DSL; this is a syntactic sugar for ${\tt Substring}(e, \lambda z. {\tt hasEntity}(z, {\tt ORG}), 1)$.}}. Thus, assuming sufficiently good neural models for keyword matching and entity recognition, the output of this program would be exactly the highlighted text for the webpages from Figure~\ref{fig:webpages}.

It is worth noting that the extraction logic described above generalizes fairly well across websites with quite different layouts. In particular, the same DSL program can be used to extract the desired information from Figure~\ref{fig:webpages2} even though this webpage looks quite different from those in Figure~\ref{fig:webpages}. 

\paragraph{Key idea \#2: Allowing imperfect solutions.} In our example so far, we were able to find a DSL program that produces exactly the highlighted text from examples in Figure~\ref{fig:webpages}. However, suppose that the pre-trained network for entity extraction is unable to recognize computer science conference names as organizations. In that case, the output of the extraction program from Eq.~\ref{eq:extractor1} would not exactly match the user-provided labels.  In fact, there is \emph{no} program in our DSL that would produce exactly the desired output.  



To deal with this difficulty, our synthesis algorithm aims to  find programs that maximize $F_1$ score  rather than looking for solutions that exactly match the user-provided labels. Thus, we frame our problem as \emph{optimal program synthesis}, where the goal is to find programs that maximize some optimization objective ($F_1$ score in our case). This optimality requirement makes the synthesis problem harder because we need to exhaustively explore the search space. 

\paragraph{Key idea \#3: Transductive learning.} An additional difficulty in our context is that there may be hundreds or even thousands of optimal solutions for a  synthesis task. In particular, given the scarcity of training examples, many different DSL programs  yield the same $F_1$ score on the labeled webpages. For instance, for the two webpages from \figref{fig:webpages}, there are actually $85$  optimal DSL programs that achieve the same $F_1$ score. 
Existing  techniques in the synthesis literature deal with the under-constrained nature of input-output examples by using heuristics to distinguish different candidate solutions. However,  standard heuristics (e.g., based on AST size) do not work well in our setting because there are still \emph{many} programs that are tied with respect to such  heuristics. 

Our approach deals with this challenge using \emph{transductive learning}. In particular, 
given all programs that yield optimal $F_1$ score on the labeled data, it generates \emph{soft labels} for unlabeled webpages by running these programs on the unlabeled webpages and aggregating their outputs. 
Then, among these programs, we choose the one whose outputs most closely match the soft labels for the unlabeled webpages.  In other words,  transductive learning  allows our method to utilize the unlabeled data to choose a most promising program and obviates the need for complex hand-crafted heuristics.

\section{Preliminaries}\label{sec:prelim}

In this section, we discuss how we represent webpages as trees.
Our   representation is different from the standard Document Object Model (DOM) and represents the nesting relationship between {text} elements on the \emph{rendered} webpage to better facilitate web question answering.

\begin{definition}{\bf (Webpage)} A \emph{webpage} is a tree $(N, E, n_0)$ with root node $n_0 \in N$, nodes $N$ and edges $E$. An edge is a pair $(n, n')$ where $n$ is the parent of  $n'$, and each node is  a triple $({\rm id}, {\rm text}, {\rm type})$ where  text is the string content of that node and ${\rm type} \in \{ \emph{list}, \emph{table}, \emph{none} \}$ indicates whether the node corresponds to an HTML list, table, or neither. 
\end{definition}

Intuitively, an edge $(n, n')$  indicates that
the text of node $n$ is the header for that of node $n'$ --- i.e., text of $n'$ is nested inside that of $n$ on the rendered version of the webpage. For instance, given an HTML document with title ``Title" and body text ``Text", our representation introduces an edge $(n, n')$ where $n$ has text ``Title'' and $n'$ contains ``Text". 
 
In our  representation, internal nodes can represent structured HTML elements like lists (both ordered and unordered) as well as tables. 
For a node $n$ representing an HTML list (resp. table), $n$'s children correspond to elements in the list (resp. rows of the table). 
\begin{example}
Our method represents the first webpage in \figref{fig:webpages} as the tree shown in \figref{fig:webtree}. 

\end{example}
\section{DSL for Web Question Answering}\label{sec:dsl}

In this section, we describe our domain-specific language called \dsl for web information extraction. At a high level, this DSL  combines pre-trained neural models for standard NLP tasks (i.e., question answering, entity extraction, and keyword matching) with symbolic constructs for manipulating strings and navigating the tree structure of the webpage. As shown in Figure~\ref{fig:dsl}, a program in this DSL takes as input a question $Q$, keyword(s) $K$, and a webpage $W$, 
and it returns a \emph{set} of strings that collectively answer the question.

\begin{figure}[t]
\footnotesize
\[
\begin{array}{rlll}

{\rm Program} \ \prog ::= & \lambda \question_, \keywordss, \webpage. \{   \guard_1 \rightarrow \lambda x. e_1, \ldots, \guard_n \rightarrow \lambda x.  e_n \} \\


{\rm Guard} \ \guard ::= & \texttt{Sat}(\nodeExt, \lambda z. \neural)  \ \ | \ \  \texttt{IsSingleton}(\nodeExt) \\


{\rm Extractor} \ \stringExt ::= & \texttt{ExtractContent}(x) \\
| & \texttt{Substring}(\stringExt, \lambda z. \neural, k) \\
| & \texttt{Filter}(\stringExt, \lambda z. \neural)\\
| & \texttt{Split}(\stringExt, c) \\

 

{\rm Section \ locator} \ \nodeExt ::= & \texttt{GetRoot}(\webpage) \\
| & \texttt{GetChildren}(\nodeExt, \lambda n.\nfilter) \\
| & \texttt{GetDescendants}(\nodeExt, \lambda n. \nfilter) \\
 
{\rm Node \  filter} \ \nfilter ::= &  \texttt{isLeaf}(n) \ | \  \texttt{isElem}(n) \\
   | &  \texttt{matchText}(n, \lambda z. \phi, b) \\
| & \top \ \ | \ \ \nfilter \wedge \nfilter \ \ | \ \ \nfilter \vee \nfilter \ \ | \ \ \neg \nfilter \\
{\rm NLP \ predicate} \ \neural ::= &  \texttt{matchKeyword}(z, \keywordss, t) \\
| &  \texttt{hasAnswer}(z, Q) \\
| &  \texttt{hasEntity}(z, l) \\
| & \top \ \ | \ \ \neural \wedge \neural \ \ | \ \ \neural \vee \neural \ \ | \ \ \neg \neural \\

\end{array}
\]
\vspace{-0.2in}
\caption{DSL for \toolname. Here, $c$ denotes a character (e.g., a  delimiter like comma), and $l$ is an entity type (e.g. Person). Also, $k \in \mathbb{Z}$, $b$ is a boolean, $t \in [0, 1]$ is a threshold.}
\label{fig:dsl}
\vspace{-0.2in}
\end{figure}

As shown in Figure~\ref{fig:dsl}, each \dsl program is a sequence of \emph{guarded expressions} of the form $\guard_i \rightarrow \lambda x. \stringExt_i$ where the \emph{guard} $\guard_i$ locates relevant tree nodes and checks whether they satisfy some property, and  the \emph{extractor} takes as input a set of tree nodes (computed by the guard) and returns a set of strings (see Figure~\ref{fig:dsl-types} for their types).
The program returns the result of  expression $\stringExt_i$ if the corresponding guard $\guard_i$ is true and all previous guards $\guard_1, \ldots, \guard_{i-1}$ evaluate to false.  Intuitively, the guards are used to determine the webpage ``schema''
and locate the relevant tree nodes from which to extract information. Then, the corresponding expression $\stringExt_i$ extracts the relevant text from those nodes. 

In more detail, guards $\guard$ in a \dsl program locate the relevant sections  $N$ of the webpage using so-called \emph{section locators} and check whether nodes $N$ satisfy some predicate. If they do, the located sections $N$ are bound to variable $x$ of the corresponding extractor expression $\stringExt$, and the result of evaluating $\stringExt$ on $N$ is returned. On the other hand, if a guard evaluates to false, then the next guard is evaluated, and this process continues until one of the guards evaluates to true. If all guards evaluate to false, the return value of the program is $\varnothing$. Next, we explain the \dsl constructs in more detail.

\begin{figure}[t]
\footnotesize
\[
\begin{array}{ll}
   \multicolumn{2}{l}{p :: \texttt{Question} \times \texttt{Keywords} \times \texttt{Webpage} \rightarrow \texttt{Set<String>}} \\
    \multicolumn{2}{l}{\guard :: \texttt{Bool} \times \texttt{Set<Node>}} \\
    \stringExt :: \texttt{Set<String>} & z :: \texttt{String} \\
    x :: \texttt{Set<Node>} & n :: \texttt{Node} \\
    \nodeExt :: \texttt{Set<Node>} & \nfilter, \neural :: \texttt{Bool}
\end{array}
\]
\vspace{-0.2in}
\caption{Types of different symbols in the \dsl grammar}
\vspace{-0.2in}
\label{fig:dsl-types}
\end{figure}

\noindent
\paragraph{Pre-trained NLP Models.} Our DSL contains three pre-trained neural models for extracting information from webpages.  These pre-trained models are used inside predicates $\neural$ and include the following primitives:
\begin{itemize}[leftmargin=*]
\item {\bf Keyword match:} Given string $z$, the  $\texttt{\small matchKeyword}(z, K, t)$ predicate evaluates to true if the semantic similarity between  $z$ and  keyword $k$ exceeds  threshold $t \in [0,1]$ for some keyword $k \in K$. 
\item {\bf Question answering:} The   $\texttt{\small hasAnswer}(z, Q)$ predicate returns true if a pre-trained neural network for textual question answering can find the answer to the given question $Q$ in input string $z$. 
\item {\bf Entity matching:} Given string $z$,  $\texttt{\small hasEntity}(z, l)$  returns true if a  neural model for entity matching decides that $z$ contains an entity of type $l$ (e.g., person, location). 
\end{itemize}
These neural primitives draw on standard NLP modeling tools for each of their respective tasks. By using standard tools, we can exploit not only pre-trained vectors \cite{pennington2014glove} and  models such as BERT \cite{bert}, but we can take advantage of training sets created for other tasks like question answering \cite{squad}. This design choice allows us to leverage neural components despite the lack of substantial training data.



\noindent
\paragraph{Section locators.} Our \dsl DSL includes so-called \emph{section locator} constructs $\nodeExt$ for identifying tree nodes from which to extract information. Section locators allow navigating the tree structure and identifying nodes that satisfy a given predicate. In particular, given a webpage $W$, $\texttt{getRoot}(W)$ returns the root node of the webpage, and the recursive {\tt getChildren} and {\tt getDescendants} constructs return respectively the children and descendant nodes satisfying a certain predicate $\nfilter$. 
Predicates on nodes allow testing whether a given node is a leaf ({\tt isLeaf}), whether it is a list/table element ({\tt isElem}), or whether the text contained in that node matches NLP predicate $\neural$ ({\tt matchText}). Note that the third boolean argument of {\tt matchText} specifies whether to consider only text within that node ($b = \texttt{false}$) or whether to consider the text in the entire subtree ($b = \texttt{true}$).

\paragraph{Guards.} As mentioned earlier, guards in our DSL are used for locating  relevant sections within a webpage and testing their properties. In particular, a guard $\guard$ uses section locators to identify relevant nodes  $N$ and then checks their properties via the {\tt IsSingleton} and {\tt  Sat} predicates.  As its name indicates, {\tt IsSingleton} tests whether $N$ contains a single node. Intuitively, this predicate is useful because existing textual question answering systems like  {\tt hasAnswer}  are more likely to be effective if the desired information can be found within a single block of text. 
On the other hand, the {\tt Sat} predicate is used to test whether \emph{any} of the nodes $n \in N$ satisfy some neural classifier $\neural$ --- i.e., ${\tt Sat}(N, \lambda z. \neural)$ checks whether  text $z$ of node $n$ satisfies $\neural$ for some $n \in N$. 


\paragraph{Extractors.} The extractor constructs are used to extract text from relevant sections  $N$ of a webpage. Note that these relevant sections are determined by the corresponding guard and bound to variable $x$ referenced in the extraction construct.  In the simplest case, the {\tt ExtractContent} function returns the string content of each node $n \in N$. The remaining constructs are recursive and allow (a) extracting substrings, (b) filtering elements from a  set, and (c) splitting a string into multiple strings. In particular, ${\tt Substring}(n, \lambda z. \neural, k)$ returns the top-k substrings satisfying neural classifier $\neural$ on $n$'s contents. Similarly, ${\tt Filter}(N, \lambda z. \neural)$ filters those nodes $n$ whose content does not satisfy $\neural$ from set $N$. Finally, ${\tt Split}(n, c)$ generates multiple new substrings by splitting $n$'s content based on the provided delimiter $c$ (e.g., comma).

\section{Optimal Neurosymbolic Synthesis} \label{sec:synthesis}

In this section, we describe our  algorithm for synthesizing  \emph{all} programs that achieve optimal $F_1$ score on a given set of training examples. At a high level, our method is based on enumerative search but employs two ideas that allow it to  scale better: First, we \emph{decompose} the task of synthesizing  extractors from that of synthesizing guards; this decomposition significantly reduces the space of programs we need to consider. Second, we exploit a certain \emph{monotonicity property} of our DSL to prune programs that are guaranteed to be sub-optimal in terms of their $F_1$ score.




Our top-level synthesis algorithm is presented in Figure~\ref{alg:top-level}. Given a few training examples $\examples$, a question $\question$, and keywords $\keywordss$, {\sc Synthesize} returns a \emph{set} of programs that achieve optimal $F_1$ score on $\examples$. At a high level, the algorithm considers all possible ways of partitioning the training examples and synthesizes optimal programs for each partition.\footnote{Since our technique only requires a small set of labeled examples, considering all partitions of $\examples$ is computationally tractable.} Intuitively, each partition corresponds to a different way of assigning guards to webpages in the training set, and the overall synthesis algorithm chooses a partition that yields  the best $F_1$ score among all partitions.

\begin{figure}[t]
\small
\vspace{-10pt}
\begin{algorithm}[H]
\begin{algorithmic}[1]
\Procedure{Synthesize}{$\examples, \question, \keywordss$}
\Statex\Input{training examples $\examples$, question $\question$, and keywords $\keywordss$.}
\Statex\Output{all \dsl programs with optimal $F_1$ score.}
\State $\res \assign \bot$; \ $opt \assign 0$;
\ForAll{$\partition \in \textsf{Partitions}(\examples)$}
\State $bs \assign []$;
\ForAll{$\examples_i \in \partition$}
\State $B \assign \textsc{SynthesizeBranch}(\examples_i, \partition \setminus \cup_{j=1}^i  \examples_j, \question, \keywordss)$;
\State $\emph{bs}.append(B )$;
\EndFor
\If{$\textsf{$F_1$}(bs, \examples) > opt$} 
\State $opt = \textsf{$F_1$}(bs, \examples)$;  $\res \assign \{ bs \} $;
\ElsIf{$\textsf{$F_1$}(bs, \examples) = opt$} 
\State$\res \assign \res \cup \{ bs \} $;
\EndIf
\EndFor
\State\Return $\res$;
\EndProcedure
\end{algorithmic}
\end{algorithm}
\vspace{-0.5in}
\caption{Top-level synthesis algorithm.}
\label{alg:top-level}
\vspace{-0.2in}
\end{figure}

In more detail, the {\sc Synthesize} procedure works as follows. It first generates all possible partitions of the training examples, and then, for  each partition $P = [ \examples_1, \ldots, \examples_n ]$, it generates a set of (optimal) programs  of the form:
\[
\psi_1  \to \lambda x. e_1, \ldots ,  \psi_n \to \lambda x. e_n  
\]
such that examples $\examples_i$  satisfy the $i$'th guard $\psi_i$ and the corresponding extractor $e_i$ achieves optimal $F_1$ score for $\examples_i$. We represent the set of optimal programs for partition $P$ as a list $bs = [B_1, \ldots, B_n]$, where each $B_i$ represents an optimal set of programs for the $i$'th branch.

In particular, a \emph{branch program} $b \in B_i$ is a pair ($\guard, \stringExt$) consisting of a guard and an extractor, and we represent a set of branch programs as a mapping $B_i$ from guards to a set of extractors $E$. 
Thus,  $B_i$ represents all  branch programs $(\guard, \stringExt)$ satisfying the following three properties:
\begin{enumerate}[leftmargin=*]
    \item The guard $\guard$ evaluates to true for all examples in $\examples_i$. 
   \item The guard $\guard$ evaluates to false for $\examples \backslash (\examples_1 \cup \ldots \cup \examples_i)$.\footnote{{Since a guard is only evaluated if previous guards evaluate to false, we only require $\psi$ to differentiate between the current set of examples and the examples that have not yet been considered.}}
   \item The extractor $\stringExt$ achieves optimal $F_1$ score for examples~$\examples_i$.
\end{enumerate}




\paragraph{Synthesizing branch programs.} Next, we consider the {\sc SynthesizeBranch} procedure (Figure~\ref{alg:branch}) for generating optimal branch programs for a given set of examples. As mentioned earlier, there are two important ideas underlying this algorithm: First, we decompose the branch synthesis problem into two separate sub-problems (one for synthesizing guards, and one for synthesizing extractors). Second, we prune the search space by inferring an upper bound on the optimal $F_1$ score that can be achieved by partial branch programs.

In more detail, the {\sc SynthesizeBranch} procedure works as follows. For a given set of positive examples $\examples^+$ and negative examples $\examples^-$, it first synthesizes a guard $\psi$ that separates $\examples^+$ from $\examples^-$ (line 4) and then generates the set of all optimal extractors using $\psi$ (line 8). Note that there may be multiple guards in our DSL that distinguish $\examples^+$ from $\examples^-$. While our algorithm considers all possible guards (loop in lines 3--12), it does so \emph{lazily} --- i.e., it only synthesizes  the next guard after synthesizing optimal extractors for the previous guards. As we will see shortly, such lazy enumeration strategy is useful because it improves the pruning power of the guard synthesis algorithm.

Now, let us consider each iteration of loop in lines 3--12. First, given a guard $\guard$ separating $\examples^+$ and $\examples^-$ (line 4), our technique infers an \emph{upper bound} on the $F_1$ score of any branch program using $\guard$ as its guard. In particular, we can do this because the extractors in our DSL are monotonic with respect to recall: If extractor $e'$ appears as a sub-expression of $e$, then the recall that can be achieved by extractor $e$ cannot be more than that of $e'$. Furthermore, since the extractor operates over the tree nodes $N$ returned by its corresponding guard, the recall can only decrease  with respect to $N$'s contents.

Our algorithm uses this observation at line 6 of  {\sc SynthesizeBranch}  by using the {\sc UB} function for computing an upper bound on branch programs using guard $\psi$. In particular, let $\nodeExt$ denote the section locator used in guard $\guard$. Then, we can obtain an upper bound for any branch program over $\guard$ using the following formula:
\small
\begin{equation}
\textsf{UB}(\nodeExt, \examples) = \frac{ 2\cdot \mathsf{Recall}(\nodeExt, \examples)}{1+ \mathsf{Recall}(\nodeExt, \examples)} \label{eq:ub}
\end{equation}
\normalsize
where  $\mathsf{Recall}(\nodeExt, \examples)$ for a section locator $\nodeExt$ and examples $\examples$ is defined as follows:
\small
\[
 \frac{\{ t  \ | \ t \in \texttt{\small ExtractContent}(\nodeExt(W)), W \in \examples_{\rm in} \} \cap \{t \ | \  t \in \examples_{\rm out}\} }{\{ t \ | \  t \in \examples_{\rm out}\}}
\]
\normalsize
where $t$ represents a token.

That is, our upper bound computation assumes  maximum possible precision (i.e., $1$) and maximum recall for any extractor using section locator $\nu$. Since $\mathsf{UB}(\nodeExt)$ gives an upper bound on the $F_1$ score of any branch program with guard $\guard$, we do not need to consider extractors for $\guard$ if   $\mathsf{UB}(\nodeExt)$ is less than the maximum $F_1$ score encountered so far (line 6).

Assuming $\guard$ is not provably sub-optimal, {\sc SynthesizeBranch} proceeds to construct optimal extractors for the synthesized guard $\guard$ (lines 7--12). To  decompose extractor synthesis from guard inference, we first compute separate input-output examples for the extractor by calling  \textsf{PropagateExamples} at line 7. In particular, this procedure executes the synthesized section locator $\nodeExt$ on the input webpages to obtain new input-output examples $\examples'$ for the extractor and invokes  {\sc SynthesizeExtractors}  on $\examples'$. 
Finally, if the branch programs associated with guard $\guard$ improve upon (or yield the same) $F_1$ score, the result set $R$ is updated.\footnote{Since branches that use guards with the same section locator have the same set of optimal extractors, the calls to {\sc SynthesizeExtractors} can be memoized across different iterations within the {\sc SynthesizeBranch} procedure. We omit this to simplify presentation.}

\begin{figure}[t]
\small
\vspace{-10pt}
\begin{algorithm}[H]
\begin{algorithmic}[1]
\Procedure{SynthesizeBranch}{$\examples^+, \examples^-, \question, \keywordss$}
\Statex\Input{Pos/neg examples $\examples^+, \examples^-$; question $\question$;  keywords $\keywordss$}
\Statex\Output{Branch programs represented as a mapping $R$ from guards to extractors  such that for each $(\psi, E) \in R$ (1) $\psi$ classifies $\examples^+, \examples^-$ and (2) $E$ achieves maximum $F_1$ score for $\examples^+$.}
\vspace{0.05in}
\State $\res \assign \bot$; $opt \assign 0$;
\While{$\text{true}$}
\State $\guard \assign \textsc{GetNextGuard}(\examples^+, \examples^-, Q, K, opt)$;
\If{$\guard = \bot$} \textbf{break};
\EndIf
\If{$\textsf{UB}(\guard.\nodeExt, \examples^+) < opt $} \textbf{continue}; 
\EndIf
\State $\examples' \assign \textsf{PropogateExamples}(\examples^+, \guard, Q, K)$;
\State $(E, F_1) \assign \textsc{SynthesizeExtractors}(\examples',  Q, K, opt)$;
\If{$F_1 > opt$} 
\State $opt \assign F_1$;  \ $\res \assign \{ (\guard, E)  \}$;
\ElsIf{$F_1 = opt$} 
\State $R[\psi] \gets E$
\EndIf
\EndWhile
\State \Return \res
\EndProcedure
\end{algorithmic}
\end{algorithm}
\vspace{-0.4in}
\caption{Algorithm for synthesizing branch programs.}
\label{alg:branch}
\end{figure}

\paragraph{Extractor synthesis.} Next, we describe the {\sc SynthesizeExtractors} procedure (Figure~\ref{alg:extractor}) for finding extractors with optimal $F_1$ score for a given set of input-output examples.  This procedure uses \emph{bottom-up} enumeration with pruning based on $F_1$ scores to reduce the search space. In particular, we use bottom-up rather than top-down enumeration because doing so allows us to  more easily exploit the monotonicity property of the DSL with respect to recall.


In more detail, {\sc SynthesizeExtractors} maintains a worklist $\worklist$ of complete extractors; and, in each iteration, it dequeues one extractor and expands it by applying all possible grammar productions for {\tt Substring}, {\tt Filter}, and {\tt Split} (line 8). A new extractor $\stringExt'$ is  added to the worklist only if $\textsc{UB}(\stringExt', \examples)$ (i.e., $F_1$ score upper bound for $\stringExt'$) is greater than or equal to the previous upper bound $s_o$ (line 9). As described earlier, we compute an upper bound on extractors generated from $e'$ by using $1$ for precision and the recall of $e'$ on the given set of examples. As before, this pruning strategy exploits the fact that if $e_1$ is a subprogram of $e_2$, then $\mathsf{Recall}(e_1, \examples) \geq \mathsf{Recall}(e_2, \examples)$ for any set of examples $\examples$. 

\begin{figure}
\small
\vspace{-10pt}
\begin{algorithm}[H]
\begin{algorithmic}[1]
\Procedure{SynthesizeExtractors}{$\examples,  \question, \keywordss, opt$}
\Statex\Input{Examples $\examples; $question $\question$;  keywords $\keywordss$}
\Statex\Input{Lower bound $opt$ on $F_1$}
\Statex\Output{Extractors $E_o$ with optimal $F_1$ score $\fonescore_{o}$ on $\examples$.}
\State $E_{o} \assign \emptyset$; \ $\fonescore_{o} \assign opt$;
\State $\worklist \assign \{ \texttt{ExtractContent}(x) \}$;
\While{$\worklist \neq \emptyset$}
\State $\stringExt \assign \worklist.remove()$; \ $\fonescore \assign \textsf{$F_1$}(\stringExt, \examples)$;
\If{$\fonescore > \fonescore_{o}$}
$E_{o} \assign \{ \stringExt \}$; $\fonescore_{o} \assign \fonescore$;
\ElsIf{$\fonescore = \fonescore_{o}$} $E_{o}.add(\stringExt)$;
\EndIf
\ForAll{$\stringExt' \in \textsf{ApplyProduction}(\stringExt)$}
\If{$\textsc{UB}(\stringExt', \examples) \geq \fonescore_{o}$}
$\worklist.add(\stringExt')$;
\EndIf
\EndFor
\EndWhile
\State\Return $(E_{o}, \fonescore_{o})$;
\EndProcedure
\end{algorithmic}
\end{algorithm}
\vspace{-30pt}
\caption{Optimal extractor synthesis.}
\label{alg:extractor}
\vspace{-0.1in}
\end{figure}

\paragraph{Lazy synthesis of guards.} The final missing piece of our synthesis algorithm is the {\sc GetNextGuard} procedure (Figure~\ref{alg:guard}) for lazy guard synthesis. In particular, this algorithm is \emph{lazy} in the sense that it yields a single guard at a time rather than returning the set of all guards separating $\examples^+$ from $\examples^-$. Since the guard synthesis algorithm also prunes its search space by computing an upper bound on $F_1$ scores, this lazy enumeration strategy improves  pruning power as the optimal $F_1$ score improves over time. However, despite the lazy nature of the guard synthesis algorithm, our technique is still guaranteed to return all optimal programs.

The guard synthesis algorithm (Figure~\ref{alg:guard})  is  similar to {\sc SynthesizeExtractors} and also performs bottom-up search with pruning. In particular, it maintains a worklist $\worklist$ of section locators. In each iteration, it dequeues one of the section locators $\nodeExt$ and generates all possible guards using $\nodeExt$ (up to some bound). If any of the resulting guards $\guard$ is a classifier between $\examples^+$ and $\examples^-$, then it is returned as the next viable guard.  Once the algorithm enumerates all possible classifiers using the section locator $\nodeExt$, it then generates all possible section locators derived from $\nodeExt$ by using the productions {\tt GetChildren} and {\tt GetDescendants} (lines 7-8). As in the previous algorithms, {\sc GetNextGuard} also computes an upper bound on  $F_1$ score  and adds a new section locator $\nodeExt'$ to the worklist if ${\sf UB}(\nodeExt', \examples^+)$ is no worse than the previous optimal $F_1$ score~$opt$ (line 8).

\begin{theorem}
Let $\examples, Q, \keywordss$ be inputs to the {\sc Synthesize} procedure and let  $p$ be a \dsl  program. Then, the set of programs returned by {\sc Synthesize}($\examples, Q, \keywordss$) includes $p$ if and only if, for any other  \dsl program $p'$,   $F_1(p)  \geq F_1(p')$. 
\end{theorem}
\label{thm:1}

\begin{figure}
\vspace{-0.15in}
\small
\begin{algorithm}[H]
\begin{algorithmic}[1]
\Procedure{GetNextGuard}{$\examples^+, \examples^-, \question, \keywordss, opt$}
\Statex\Input{Pos/neg examples $\examples^+,\examples^-$}
\Statex\Input{Question $\question$; keywords $\keywordss$;  lower bound $opt$ on $F_1$ score}
\Statex\Output{Next guard that classifies $\examples^+$ and $\examples^-$.}
\vspace{0.05in}
\State $\worklist \assign \{ \texttt{GetRoot}(\webpage) \}$;
\While{$\worklist \neq \emptyset$}
\State $\nodeExt \assign \worklist.remove()$;
\ForAll{$\guard \in \textsf{GenGuards}(\nodeExt)$}
\If{$\forall \examplee_{in} \in \examples^+. \semantics{\guard(\examplee_{in},\question,\keywordss)}$ 
$\wedge$
\Statex \ \ \ \ \ \ \ \ \ \ \ \ \ \ \ \ \ \ \ \ \ \ $\forall \examplee_{in} \in \examples^-. \neg \semantics{\guard(\examplee_{in},\question,\keywordss)}$} \textbf{yield} $\guard$;
\EndIf
\EndFor
\ForAll{$\nodeExt' \in \textsf{ApplyProduction}(\nodeExt)$}
\If{$\textsf{UB}(\nodeExt', \examples^+) \geq opt$}
$\worklist.add(\nodeExt')$;
\EndIf
\EndFor
\EndWhile
\State \Return $\bot$;
\EndProcedure
\end{algorithmic}
\end{algorithm}
\vspace{-30pt}
\caption{Lazy enumeration of guards.}
\label{alg:guard}
\vspace{-0.15in}
\end{figure}

\section{Program Selection via Transductive Learning}~\label{sec:self-supervision}

Our optimal synthesis algorithm outputs all programs that have optimal $F_1$ score on the labeled training data. However, 
not all of these programs generalize well to new inputs. 
In this section, we present a technique based on transductive learning \cite{zhu05survey} that selects a program that generalizes well beyond the training examples. 

We describe our algorithm (summarized in Figure~\ref{alg:selects}) for selecting a program that generalizes well to the test set. At a high level, our selection method must satisfy two objectives. First, it should select a program that generalizes well in terms of $F_1$ score. Empirically, we observe that a large fraction of optimal programs achieve  good $F_1$ score on the test set, so a randomly chosen program has good $F_1$ score on average. However, we also want to minimize variance---in many cases, a sizable fraction of programs perform quite poorly, so a randomly chosen program may have  poor $F_1$ score. Our selection method is designed to select a good program while avoiding these poorly performing programs.

The key concept underlying our approach is to use an \emph{ensemble} of the optimal programs $\Pi^*$ generated by our synthesis algorithm; that is, we aggregate predictions over a large random sample of optimal programs. Such an ensemble would address both of the above points. Ensembles of multiple models typically generalize better than the individual models since the errors made by individual models tend to average out~\cite{opitz1999popular}. For the same reason, they also tend to reduce the variance in performance~\cite{dietterich2000ensemble}. However, directly using an ensemble instead of an individual program has a few drawbacks. First,  an ensemble is significantly less interpretable than an individual model. Second, since an ensemble includes many programs, there is a large computational cost to using the ensemble  if the learned model is to be used over and over again.

%

\begin{figure}[t]
\vspace{-10pt}
\small
\begin{algorithm}[H]
\begin{algorithmic}[1]
\Procedure{Select}{$\examples, \inputexamples, \Pi^*$}
\Statex\Input{training examples $\examples$, unlabeled input examples $\inputexamples$, optimal programs $\Pi^*$.}
\Statex\Output{optimal program $\pi^*$ on the transductive learning objective.}
\State draw i.i.d. samples $\pi_1,...,\pi_{\mathcal{N}}\sim\Pi^*$;
\State construct ensemble $\Pi_E=\{\pi_1,...,\pi_{\mathcal{N}}\}$;
\State compute the output $\outputexamples_j\in\Pi_E$ for $\pi_j$ according to Eq.~\ref{eqn:ensembleoutputs};
\State compute $L(\pi)=\sum_{j=1}^{\mathcal{N}}L(\pi;\inputexamples,\outputexamples_j)$ for each $\pi\in\Pi_E$;
\State \Return $\pi^*=\operatorname*{\arg\max}_{\pi\in\Pi_E}L(\pi)$;
\EndProcedure
\end{algorithmic}
\end{algorithm}
\vspace{-35pt}
\caption{Program Selection Algorithm.}
\label{alg:selects}
\end{figure}

Thus, our algorithm first builds the ensemble and then \emph{compresses} it into a single program by leveraging  the \emph{unlabeled} training data. In particular, it constructs an ensemble $\Pi_E$ by sampling $\mathcal{N}$ optimal programs returned by the synthesis algorithm (lines 2--3). Then, it uses the ensemble to generate soft labels for the unlabeled webpages and returns the program $\pi^*$ that minimizes loss $L(\pi^*)$ with respect to these soft labels. We describe our approach in more detail below. For conciseness, we give a high-level sketch of our derivations, and provide details in the Appendix.



\paragraph{Transductive learning objective.}
Given (i) labeled examples $\examples$, (ii) unlabeled inputs $\inputexamples$, and (iii) optimal programs $\Pi^*$ returned by  {\sc Synthesize}, our algorithm finds a program $\pi \in \Pi^*$ that minimizes the following objective:
\begin{align}
\label{eqn:selfsupervisedobj}
\tilde{L}(\pi;\examples,\inputexamples)=\mathbb{E}_{p(\outputexamples\mid\inputexamples,\examples)}[L(\pi;\inputexamples,\outputexamples)].
\end{align}
The expression $L(\pi;\inputexamples,\outputexamples)$ is a loss function we wish to minimize in a standard supervised learning fashion using $(\inputexamples, \outputexamples)$ as the training dataset---e.g., we could take it to be the negative $F_1$ score. However, the difficulty is that we do not know the labels $\outputexamples$ for the inputs $\inputexamples$. Thus, we take the expectation with respect to a distribution $p(\outputexamples\mid\inputexamples,\examples)$ that leverages information from the labeled examples. As described in detail below, this distribution is constructed by using an ensemble of programs synthesized based on $\examples$ to assign soft ``pseudo-labels'' to $\inputexamples$.


\paragraph{{Generating} labels via program ensembling.}

Next, we describe how to construct the distribution $p(\outputexamples\mid\inputexamples,\examples)$. We do this in two steps: first, by defining a distribution $p(\pi' \mid \examples)$ over (optimal) programs conditioned on the input data, then using the fact that these programs are deterministic to compute $p(\outputexamples \mid \pi', \inputexamples)$.

We define our distribution $p(\pi'\mid\examples)$ to assign probability mass \emph{only} to programs in $\Pi^*$, those that best satisfy the given examples $\examples$.\footnote{We note that other choices are possible (e.g., prioritizing smaller programs, including some probability on erroneous programs, etc.); we found this choice to work well empirically.}
Ideally, we could use
the uniform distribution over optimal programs $\Pi^*$. However, a key difficulty is that summing over all programs $\pi\in\Pi^*$ is intractable since the cardinality of $\Pi^*$ is too large in practice. Thus, we instead approximate this distribution by constructing an ensemble $\Pi_E=\{\pi_1,...,\pi_{\mathcal{N}}\}$, where $\pi_i\sim\text{Uniform}(\Pi^*)$ are i.i.d. samples and where $\mathcal{N} \in \mathbb{N}$ is a hyperparameter, and then using
\small
\begin{align}
\label{eqn:selfsuperviseddist}
p(\pi\mid\examples)=\frac{\mathbbm{1}(\pi\in\Pi_E)}{\mathcal{N}}
\end{align}
\normalsize

Finally, given this distribution, we have
\small
\begin{align}
\label{eqn:derivation}
p(\outputexamples\mid\inputexamples,\examples)
&= \sum_{\pi'\in\Pi}p(\pi'\mid\examples)\cdot p(\outputexamples\mid\pi',\inputexamples) \\
&= \sum_{\pi'\in\Pi}p(\pi'\mid\examples)\cdot\prod_{k=1}^K \mathbbm{1}(\outputexamplee_k=\pi'(\inputexamplee_k)) \nonumber
\end{align}
\normalsize
where the second step follows from the fact that our programs are deterministic and each place probability 1 over a single output.

\paragraph{Program selection.}

Finally, our algorithm aims to select
\small
\begin{align}
\label{eqn:selectedprog}
\pi^*=\operatorname*{\arg\min}_{\pi\in\Pi_E}\tilde{L}(\pi;\examples,\inputexamples).
\end{align}
\normalsize
i.e., the program that minimizes the loss with respect to the ensemble. 
We now describe how to evaluate $\tilde{L}(\pi;\examples,\inputexamples)$. 

First, we precompute the possible outputs
\small
\begin{align}
\label{eqn:ensembleoutputs}
\outputexamples_j=(\pi_j(\inputexamplee_1),...,\pi_j(\inputexamplee_K))\qquad(\forall\pi_j\in\Pi_E),
\end{align}
\normalsize
in which case we have
\small
\begin{align*}
p(\outputexamples\mid\inputexamples,\examples)
&=\frac{1}{\mathcal{N}}\sum_{\pi'\in\Pi_E}\prod_{k=1}^K\mathbbm{1}(\outputexamplee_k=\pi'(\inputexamplee_k)) \\ &=\frac{1}{\mathcal{N}}\sum_{j=1}^{\mathcal{N}}\mathbbm{1}(\outputexamples=\outputexamples_j).
\end{align*}
\normalsize
In other words, when evaluating $p(\outputexamples\mid\inputexamples,\examples)$, we only need to account for outputs $\outputexamples_j$ according to programs $\pi_j\in\Pi_E$. Thus, we have:
\small
\begin{align}
\tilde{L}(\pi;\examples,\inputexamples)
&=\sum_{\outputexamples}p(\outputexamples\mid\inputexamples,\examples)\cdot L(\pi;\inputexamples,\outputexamples)\\  &=\frac{1}{\mathcal{N}}\sum_{j=1}^{\mathcal{N}}L(\pi;\inputexamples,\outputexamples_j). \label{eqn:derivationgoal}
\end{align}
\normalsize
Substituting into Eq.~\ref{eqn:selectedprog}, our algorithm selects the program
\begin{align}
\label{eqn:selectedprog2}
\pi^*=\operatorname*{\arg\min}_{\pi\in\Pi_E}\sum_{j=1}^{\mathcal{N}}L(\pi;\inputexamples,\outputexamples_j),
\end{align}
which is equivalent to Eq.~\ref{eqn:selectedprog} since $\mathcal{N}$ is a positive constant.

\section{Implementation}\label{sec:impl}

 In this section, we provide  implementation details about different components of \toolname.

\paragraph{Parsing.} As explained in Section~\ref{sec:prelim}, \toolname represents each webpage as a tree that captures relationships between text elements on the renderd version of the webpage. Thus, \toolname first parses a given HTML document into our internal representation. To do this, we first extract the DOM tree representation using the \emph{BeautifulSoup4} HTML parser and remove unnecessary elements such as images and scripts. Then, when converting to our tree representation, we follow the standard HTML header hierarchy. In particular, the H1 header corresponds to the root node, and $H_{i+1}$ headers are represented as children nodes of the $H_i$ headers. 




\paragraph{Interactive labeling.} Rather than asking users to directly provide labeled webpages, \toolname
 actually interacts with users and  suggests webpages to label. The goal here is to minimize the number of user annotations while ensuring that the labels achieve good coverage of different schemas in the test set. To do so, \toolname clusters webpages based on various features, including which section locator constructs in our DSL yield non-empty answers, the type of entities contained in the extracted sections, the layout of extracted sections etc. We then identify webpages that are similar to and different from the webpages labeled so far and ask the user to label these additional webpages. In practice, we restrict the number of user queries to at most five.

\paragraph{Neural modules in the DSL} Our tool leverages several existing natural language processing frameworks and models to implement the neural modules. For QA-related constructs, we use the BERT QA system  \cite{bert} as the underlying model. Specifically, we use the version that has been fine-tuned on the SQUAD dataset\footnote{The link to the model: https://huggingface.co/bert-large-uncased-whole-word-masking-finetuned-squad. }\cite{squad}. We use Sentence-BERT~\cite{sbert} to generate sentence embeddings for keyword similarity,  and we employ  Spacy\footnote{https://spacy.io/. Specifically, we use the ``en\_core\_web\_md'' model.}\cite{spacy2} for named entity extraction and sentence segmentation. Since the keyword matching module requires a real-valued threshold $t \in [0,1]$, our synthesis algorithm discretizes  it  using a step size of $0.05$.

\paragraph{Transductive learning loss. } {Recall that our program selection technique proposed in Section~\ref{sec:self-supervision} is parametrized over a loss function $L(\pi; \mathcal{I}, \mathcal{O})$. In our implementation, we take our loss function to be the Hamming distance between the sets of words extracted by each program. In particular, we use loss function: $L(\pi; \mathcal{I}, \mathcal{O}) = {\tt Hamming}(\pi(\mathcal{I}), \mathcal{O}).$}

\paragraph{Hyperparameters.} The {\toolname} system has certain hyper-parameters that control the maximum depth of synthesized programs. By default, the hyper-parameter for guard depth is set to $7$ and the one for extractor depth is set to be $5$.  There is also another hyper-parameter (with a default value of $1000$) that controls the number of programs used to construct an ensemble during transductive learning.

\section{Evaluation}~\label{sec:eval}
In this section, we describe a series of experiments that are designed to answer the following research questions:

\begin{itemize}[leftmargin=*]
    \item {\bf RQ1.} How does \toolname's performance compare against  other question answering and information extraction tools?
    \item {\bf RQ2.} How important are the synthesis techniques proposed in Section~\ref{sec:synthesis}?
\item {\bf RQ3.} Is the  program selection technique based on transductive learning (Section~\ref{sec:self-supervision}) useful in practice?
\end{itemize}

\paragraph{Benchmarks} To answer these questions, we evaluate \toolname on 25 different tasks across four different domains, namely   faculty profiles, computer science conferences,  university courses, and clinic websites. For each domain, we collect approximately 40 webpages and evaluate the performance of each tool in terms of $F_1$ score, precision, and recall. For each task, out of around 40 webpages, around 5 of them are used for training (i.e. synthesis) and the remaining is the test set. Table~\ref{tab:task} describes the 25 tasks used in our evaluation.

\paragraph{Experimental Setup} All of our experiments are conducted
on a machine with Intel Xeon(R) W-3275 2.50 GHz CPU
and 16GB of physical memory, running the Ubuntu 18.04
operating system with a NVIDIA Quadro RTX8000 GPU. 

\begin{table}[]
    \centering
    \scriptsize
    \caption{Description of the tasks used in evaluation.}
    \begin{tabular}{|c|l|}
    \hline
    Domain & Description  \\
    \hline
    \multirow{8}{*}{Faculty} & Extract  current PhD students \\
     & Extract  conference publications at PLDI  \\
     & Extract  courses they have taught   \\
     & Extract those papers that received a Best Paper Award \\
     & Extract  program committees they have served on \\
     & Extract  conference papers  they published in 2012 \\
     & Extract  co-authors among all papers published at PLDI \\
     & Extract  formerly advised students \\
     \hline
     \multirow{6}{*}{Conference} & Extract  program committee members \\
     & Extract  program chairs \\
     & Extract the topics of interest  \\
     & Extract the paper submission deadlines \\
     & Extract whether the conference is single-blind or double-blind \\
     & Extract  institutions PC members are from \\
     \hline
     \multirow{6}{*}{Class} & Extract the name of instructors \\
     & Extract the time of the lectures \\
     & Extract the name of teaching assistants \\
     & Extract the date of the exams \\
     & Extract  information about  textbooks \\
     & Extract information on how   grades are assigned \\
     \hline
     \multirow{5}{*}{Clinic} & Extract the doctors or providers \\
     & Extract the provided services \\
     & Extract the types of treatments they specialize in\\
     & Extract the accepted insurances \\
     & Extract the locations \\
     \hline
    \end{tabular}
    \vspace{-10pt}
    \label{tab:task}
\end{table}

\subsection{Comparison with Other Tools}

To answer our first research question, we compare \toolname\ against the following baselines:
\begin{itemize}[leftmargin=*]
\item {\sc BERTQA}~\cite{bert}: This is a state-of-the-art textual question answering system that takes as input an entire webpage and a question and outputs the answer.~\footnote{We also tried fine-tuning this model using the labels in our training examples; however, we do not report results for the fine-tuned model since 
its result is
actually worse compared to ~\cite{bert}.}. 
\item {\sc HYB}~\cite{raza2020web}: This is a programming-by-example system that takes a set of webpages as input and synthesizes XPath programs for data extraction. 
\item {\sc EntExtract}~\cite{liang2014}: This is a zero-shot entity extraction tool for webpages using a natural language query as input. 
\end{itemize}

Note that these baselines do not address \emph{exactly} the same problem addressed by \toolname in that they take fewer inputs. Thus, while our comparison is not completely apples-to-apples, these systems are the closest ones to \toolname for performing a comparison.
 
\begin{figure}
    \centering
    \definecolor{independence}{RGB}{121,125,98}
\definecolor{heliotropegray}{RGB}{154,140,152}
\definecolor{silverpink}{RGB}{201,173,167}
\definecolor{isabelline}{RGB}{242,233,228}

\definecolor{usafablue}{RGB}{0,78,137}
\definecolor{orangecrayola}{RGB}{255,107,53}
\definecolor{peachcrayola}{RGB}{247,197,159}
\definecolor{beige}{RGB}{239, 239, 208}

\definecolor{darkbluegray}{RGB}{102,106,134}
\definecolor{shadowblue}{RGB}{120,138,163}
\definecolor{opal}{RGB}{146,182,177}
\definecolor{laurelgreen}{RGB}{178,201,171}
\definecolor{dutchwhite}{RGB}{232,221,181}
\begin{tikzpicture}[scale=0.9]
\begin{axis}[
    ymin=0,
    ymax=0.8,
    ybar,
    enlarge x limits=0.20,
    legend style={at={(0.5,-0.15)},
      anchor=north,legend columns=-1},
    ylabel={Avg score},
    symbolic x coords={$F_1$,Precision,Recall},
    xtick=data,
    legend image code/.code={
        \draw [#1] (0cm,-0.1cm) rectangle (0.2cm,0.25cm); },
    ]
\addplot[black,fill=darkbluegray
] coordinates {($F_1$, 0.70) (Precision, 0.69) (Recall, 0.73)};
\addplot[black,fill=opal
] coordinates {($F_1$, 0.25) (Precision, 0.47) (Recall, 0.17)};
\addplot[black,fill=peachcrayola
] coordinates {($F_1$, 0.05) (Precision, 0.34) (Recall, 0.03)};
\addplot[black,fill=dutchwhite
] coordinates {($F_1$, 0.08) (Precision, 0.06) (Recall, 0.15)};
\legend{\toolname,{\sc BERTQA},{
    \sc HYB}, {\sc EntExtract}}
\end{axis}
\end{tikzpicture}
    \caption{Comparison between \toolname and other tools }
    \vspace{-10pt}
    \label{fig:baselines}
\end{figure}
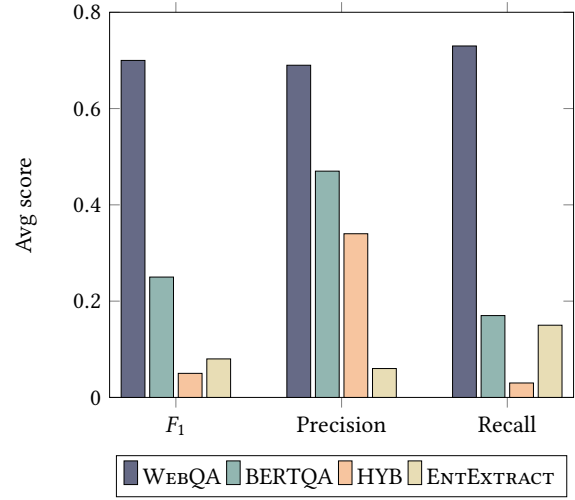

\begin{table*}[]
\small
    \centering
    \caption{Evaluation results for each baseline per domain. P stands for Precision and R means Recall.}
    \begin{tabular}{|c|ccc|ccc|ccc|ccc|}
    \hline
    \multirow{2}{*}{Domain} & \multicolumn{3}{c|}{\toolname} & \multicolumn{3}{c|}{{\sc BERTQA}} & \multicolumn{3}{c|}{{\sc HYB}} & \multicolumn{3}{c|}{{\sc EntExtract}} \\
    & P & R & $F_1$ & P & R & $F_1$ & P & R & $F_1$ & P & R & $F_1$ \\
    \hline
    Faculty & 0.72 & 0.80 & 0.75 & 0.44 & 0.08 & 0.18 & 0.48 & 0.02 & 0.04 & 0.02 & 0.14 & 0.04 \\
    Conference & 0.71 & 0.69 & 0.70 & 0.58 & 0.31 & 0.32 & 0.26 & 0.02 & 0.03 & 0.07 & 0.20 & 0.09 \\
    Class & 0.63 & 0.77 & 0.68 & 0.55 & 0.26 & 0.31 & 0.18 & 0.04 & 0.04 & 0.04 & 0.09 & 0.05 \\
    Clinic & 0.71 & 0.62 & 0.66 & 0.31 & 0.02 & 0.04 & 0.42 & 0.06 & 0.09 & 0.14 & 0.20 & 0.16 \\
    \hline
    \end{tabular}
    \label{tab:domain-results}
\end{table*}

Our main results are summarized in \figref{fig:baselines}, and Table~\ref{tab:domain-results} shows a more detailed breakdown of results across our four domains. As we can see from \figref{fig:baselines}, \toolname outperforms all three baselines in terms of average $F_1$ score, precision, and recall, and, accoring to Table~\ref{tab:domain-results}, these results hold across all four domains.  Among the three other tools, {\sc BertQA} has the best performance; however, it has significantly worse recall and $F_1$ score compared to \toolname. 

\paragraph{Failure analysis for the baselines.} We briefly explain why the baseline systems perform poorly in our evaluation.  As mentioned earlier, a textual QA system like  {\sc BERTQA} fails to take advantage of the inherent structure in webpages and performs particularly poorly on tasks that require extracting multiple different spans from the input webpage. On the other hand, {\sc EntExtract} does leverage the tree structure of the webpage but we found that it often returns irrelevant answers (e.g., publications instead of students). We believe this is because {\sc EntExtract} generates extraction predicates based on XPath queries, but most of our tasks are difficult to solve using simple XPath programs.  Finally, {\sc HYB} tries to synthesize programs that exactly match the provided labels (i.e., perfect $F_1$ score); however, since such programs do not exist for many tasks, synthesis fails in several cases. 



\paragraph{Failure analysis for \toolname} There are two tasks on which \toolname does not significantly outperform the {\sc BERTQA} baseline. One of these tasks is extracting conference submission deadlines, and the other one is determining whether a conference is double-blind or not. For these two tasks, the program synthesized by \toolname essentially returns the output of the QA model; hence, it does not outperform {\sc BERTQA}.  



\subsection{Evaluation of the Synthesis Engine}
In this section, we describe an ablation study that quantifies the impact of the proposed synthesis techniques from Section~\ref{sec:synthesis}. In particular, recall that our synthesis algorithm incorporates two key ideas--decomposition and pruning based on $F_1$ score. To evaluate the relative importance of these ideas, we consider the following two ablations of \toolname:
\begin{itemize}[leftmargin=*]
    \item \toolnameD: This is a variant of \toolname that synthesizes guards and extractors jointly. In other words, it does not decompose the synthesis problem into two separate guard synthesis and extractor synthesis sub-tasks. 
    \item \toolnameP: This variant does not compute an upper bound on the $F_1$ score of partial programs. Thus, it is unable to prune partial programs from the search space. 
\end{itemize}

The results of this ablation study are presented in \figref{tab:pruning}. Here, the first column shows average synthesis time (in seconds) for all three variants, and the second column shows the average speedup of \toolname over its two ablations. As we can see, both decomposition and $F_1$-based pruning play a significant role in improving synthesis time.  In particular, pruning improves synthesis time by a factor of $3.6$ and decomposition improves it by a factor of $2.4$.  Note that we do not report $F_1$ scores in Table~\ref{tab:pruning} since all variants synthesize the same programs but  differ in how long they take to do so. 


\begin{table}[]
    \small
    \centering
    \caption{Results of the ablation  study. This table shows the average training time and the average speedup that \toolname achieves compared to the other two techniques. }
    \begin{tabular}{|c|c|c|}
    \hline
    Technique & Avg time (s) & Avg Speedup \\ 
    \hline
     \toolname & 419 & -  \\
    {\small \toolnameP} & 1351 & 3.6 \\ 
   {\small   \toolnameD} & 931 & 2.4 \\ 
    \hline
    \end{tabular}
    \label{tab:pruning}
    \vspace{-0.2in}
\end{table}

\subsection{Effectiveness of the Transductive Learning}
In this section, we investigate the usefulness of the transductive learning technique from Section~\ref{sec:self-supervision} by comparing against two simpler baselines:

\begin{itemize}[leftmargin=*]
    \item {\sc Random}: This baseline chooses uniformly at random one of the optimal programs for the training examples. 
    \item {\sc Shortest}: This baseline   chooses uniformly at random one of the \emph{smallest} programs  (in terms of AST size) that  optimize $F_1$ score on the training examples. 
    
    
\end{itemize}

Recall that the transductive learning technique from Section~\ref{sec:self-supervision} is both intended to reduce variance and produce better-quality results on the test set. Thus, we compared \toolname against two baselines in terms of the following two metrics\footnote{In the experiment, these two metrics are computed based on 20 runs.}.

\begin{itemize}[leftmargin=*]
    \item {Mean:} We report percentage improvement of the transductive learning technique in terms of average $F_1$ score over the two baselines.
    \item {Variance:} We also report  the average reduction in variance.

\end{itemize}

As we can see from Table~\ref{fig:self-supervised}, the transductive learning technique dramatically reduces variance and modestly improves average $F_1$ score. Thus, by using our proposed transductive learning technique, \toolname achieves much more stable performance (in terms of the quality of the synthesized programs) compared to these other approaches. 

\begin{table}[]
    \small
    \centering
    \caption{Evaluation of  transductive learning. This table shows the \% of improvement in $F_1$ and the reduction in variance of \toolname  compared to the other two techniques. }
    \label{fig:self-supervised}
    \begin{tabular}{|c|c|c|}
    \hline
      Technique   &  \% Improvement in $F_1$ & Reduction in Variance  \\
    \hline
    {\sc Random} & 6.0\% & 1550$\times$ \\
    {\sc Shortest} & 6.3\% & 1570$\times$ \\
    \hline
    \end{tabular}
\end{table}

\noindent 
\paragraph{Remark.} Appendix C presents additional ablation studies evaluating the impact of the different  input modalities as well as the number of training examples. 

\section{Related Work}\label{sec:related}

\paragraph{Program synthesis for webpages}

There is  a large body of prior research on learning extraction rules from HTML documents.
In data mining, this problem is known as \emph{wrapper induction}~\cite{kushmerick1997wrapper}, 
and there is a wide spectrum of proposed solutions~\cite{muslea1999hierarchical, hsu1998generating, chang2001iepad, crescenzi2001roadrunner, anton2005xpath, raza2020web, flashextract, gulhane2011web}.  
For instance, \textsc{Vertex}~\cite{gulhane2011web} uses an \emph{apriori} style algorithm~\cite{agarwal1994fast} to learn XPath-based rules from human annotated sample pages, and ~\cite{anton2005xpath} also  learns XPath-compatible wrappers.
While these techniques can extract HTML elements, they cannot perform finer-grained string processing \emph{inside} HTML elements.
In contrast, \textsc{FlashExtract}~\cite{flashextract} can perform some text manipulation inside HTML elements\footnote{{We were not able to experimentally compare against {\sc FlashExtract} because their released implementation in Prose does not support text manipulation in HTML elements.}}; however, unlike our approach, it does not use neural constructs, making it difficult to apply this technique to structurally heterogeneous websites.
Recent work~\cite{iyer2019synthesis} targets data extraction from heterogeneous sources by combining ideas from program synthesis and machine learning. However, this approach requires significant number of \emph{manually} labeled samples since it relies on first training an ML model. In contrast, our technique uses pre-trained models and a small amount of training data.

Recent work by \citet{raza2020web} proposes {\sc HYB}, a synthesis-driven web data extraction technique that is now deployed in the Microsoft Power BI product.
This technique is also based on program synthesis and uses a combination of top-down and bottom-up search. As shown in our evaluation, \toolname performs significantly better than their approach; we believe this is due to the fact that our method is based on a neurosymbolic DSL. 

There is also a line of work~\cite{chasins2018rousillon, barman2016ringer, chasins2015browser, chasins2017skip} for learning web automation macros using a programming-by-demonstration approach. These techniques perform scraper synthesis by recording user interactions with a few webpages and then generalize these interactions into a programmatic webpage scraper. In contrast to our approach, these techniques target structurally similar pages (e.g., different Amazon products) and use a different type of input, namely demonstrations.




\paragraph{Information extraction from text} Much of the IE literature (e.g., relation extraction  \cite{mintz09distant}) is confined to a given database schema. Among IE frameworks,   ``slot-filling'' approaches   \cite{freitag00,patwardhan2007} typically require at least medium-sized training sets to work on  specific schemas, 
and most few-shot approaches \cite{HanEtAl2018,SoaresEtAl2019} use a pre-trained model and perform further training on text that expresses the desired relations in a similar fashion to the target domain.
In contrast, open information extraction techniques\cite{etzioni08openie} aim to retrieve data in an ontology-free way that can theoretically be used for downstream tasks like question answering \cite{choi15scalable}. However, this information is extracted primarily from textual relations rather than structured formatting; even graph-based approaches use graphs over textual relations only \cite{qian19graphie}. Therefore, these approaches do not work well in settings (such as ours) that 
involve a combination of tree structure and free-form text. 

\paragraph{Information extraction from semi-structured data} Recent work has begun to tackle the problem of semi-structured data, particularly interactions between tables and natural language. Prior work looks at extracting lists from the web \cite{liang2014}, answering questions from tables \cite{pasupat15wikitableq}, verifying facts from tables \cite{chen19tabfact}, or generally extracting information from tables \cite{wu18fonduer}. However, much of this work assumes access to large training sets or relies heavily on the structure of tables.

Two recent efforts have tackled the problem of IE from semi-structured data in a setting similar to ours \cite{lockard20zeroshotceres,lin20freedom}. However, to use these techniques in our setting, we would have to first run their tools
to extract a knowledge base, and then interpret our questions into some kind of semantic representation that we can execute against the extracted knowledge base. In contrast to such an approach, our work instead \emph{dynamically} learns the relation to extract from the question, specified keywords, and examples. 

\paragraph{Question answering} Beyond the table-based question answering approaches listed above, there is little work on question answering over text that can be directly applied to our setting. BERT-based \cite{bert} models applied to datasets such as SQuAD \cite{squad} only work well on input that is completely unstructured text. While there are some recent efforts on question answering with more programmatic structures \cite{gupta20neural} for tasks like DROP \cite{dua19drop}, these systems are highly specialized to applications like answering numerical questions.

\paragraph{Quantitative program synthesis}
There has also been recent work on optimal program synthesis with respect to a quantitative objective. 
For example,~\citet{metasketch} introduce a general framework for optimal program synthesis where the search space is represented by a set of sketches. Their technique uses the objective function together with a gradient function to direct the search. In contrast to~\citet{metasketch}, our work specifically targets the web question answering domain, uses a neurosymbolic DSL, and employs an objective function that is based on program semantics. While ~\citet{metasketch}, in principle, also support semantic objective functions, they require the objective function to be reducible to a decidable theory, which  does not hold in our case due to the use of neural primitives.
\textsc{QuaSi}~\cite{hu2018syntax} also considers the problem of synthesis with quantitative objectives, but it requires the objectives to be syntactic. 
Other existing synthesis techniques~\cite{flashfill, wang2017program, wang2017synthesis},  are mostly \emph{qualitative}, although they implicitly use a ranking function as an inductive bias to help with generalization. However, such ranking functions are quite restricted and  mostly syntactic (e.g., based on program size).

\paragraph{Neurosymbolic DSLs.}

There has been recent interest in neurosymbolic DSLs that include both logical and neural components. For instance, neural module networks~\cite{andreas2016neural,andreas2016learning} dynamically compose DNNs for tasks like predicting object attributes in images~\cite{mao2019neuro} or identifying numbers  in text~\cite{gupta20neural}. However, these techniques use purely neural components (even for operations like  filtering and counting), which significantly increases sample complexity. Recent approaches have trained combinations of neural and logical components  by backpropagating through such programs~\cite{gaunt2017differentiable,valkov2018houdini,shah2020learning}. There has also been work on synthesizing neurosymbolic programs to represent structure in images~\cite{ellis2018learning,young2019learning} and reinforcement learning policies~\cite{inala2020learning,anderson2020neurosymbolic}. Overall, existing approaches largely focus on simultaneously learning the program structure and the DNN parameters. Hence, they are limited to very simple programs and neural components, as they need to optimize  neural network parameters using backpropagation. In contrast, our work is designed to incorporate state-of-the-art DNNs such as BERT, which take significant time to train. In addition, we search over tens of thousands of programs by relying on pretrained DNN models and by developing novel deduction techniques for optimal synthesis.

\paragraph{Multi-modal program synthesis.} There has been growing interest in program synthesis from multiple modalities of specifications. For instance, several works have used a combination  of natural language  and input-output examples to synthesize regular expressions, data wrangling and string manipulation programs,  SQL queries, and temporal logic formulas \cite{regel, mars, raza2015, duoquest, gavraninteractive}. Our technique can also be viewed as an instance of multi-modal synthesis that is based on a neurosymbolic programming language.


\section{Conclusion}\label{sec:concl}

We have presented \toolname, a new synthesis-powered system for extracting information  from webpages. 
We have evaluated \toolname on 25 different tasks spanning four different domains and 160 different webpages and show that \toolname significantly outperforms competing approaches in terms of $F_1$ score, precision, and recall.

\begin{acks}                            
  We thank our shepherd Uri Alon as well as our anonymous reviewers and members of the UToPiA group for their helpful feedback. This material is based upon work supported by the
  \grantsponsor{GS100000001}{National Science Foundation}{http://dx.doi.org/10.13039/100000001} under Grant
  No. CCF-\grantnum{GS100000001}{1811865}, Grant No. CCF-\grantnum{GS100000001}{1762299} and Grant No. CCF-\grantnum{GS100000001}{1918889}.
\end{acks}

\balance
\bibliography{main}

\newpage
\appendix
\section{Proofs}

\begin{definition}(Extension of a sub-program) A sub-program $p'$ is an extension of a sub-program $p$  in the \dsl DSL if $p' \in {\tt ApplyProduction}(p)$. 

\end{definition}

\begin{lemma}\label{lem:ub}(Correctness of the {\sc UB} (Equation~\ref{eq:ub})) Given a set of examples $\examples$, for any \dsl sub-program $p$, ${\sc UB}(p, \examples) \geq F_1(p, \examples)$. 
\end{lemma}

\begin{proof}
Recall that the $F_1$ score is computed as:
\[
F_1(p, \examples) = \frac{2 \cdot {\tt Precision}(v, \examples) \cdot {\tt Recall}(v, \examples)}{{\tt Precision}(v, \examples) + {\tt Recall}(v, \examples)}
\]

where $max({\tt Precision}(v, \examples)) = 1.0$. According to the  {\sc UB} computation, we always set ${\tt Precision}(v, \examples)$  to 1.0, ${\sc UB}(p, \examples) \geq F_1(p, \examples)$.
\end{proof}

\begin{theorem}\label{thm:pruning} (Correctness of the pruning procedure (i.e. line 9 in \figref{alg:extractor}, line 8 in \figref{alg:guard} and line 6 in \figref{alg:branch})) Given a DSL subprogram $p$ and a set of examples $\examples$, let $p'$ be any valid extension of $p$ in the \dsl DSL, then $F_1(p', \examples) \leq {\sc UB}(p, \examples)$.

\end{theorem}

\begin{proof}[Proof Sketch] Recall that given a subprogram $p$ and $p'$ that is instantiated by a valid production in the DSL that takes $p$ as its argument, the output of $p'$ is a subset of tokens in the output of $p$ on the same given example $e$. Therefore, let ${\tt Recall}(p, e)$ represents the recall of the output of program $p$ with respect to the ground truth of example $e$,  we have  ${\tt Recall}(p, e) \geq {\tt Recall}(p', e)$. According to the upper bound computation in Equation~\ref{eq:ub}, we fix the precision  of any output to be the best scenario, $1.0$. Therefore, ${\sc UB}(p, e) \geq {\sc UB}(p', e)$. Using Lemma~\ref{lem:ub}, we obtain that $F_1(p', e) \leq {\sc UB}(p, e)$. Without the loss of generality, given a set of examples $\examples$, we have  ${\sc UB}(p, \examples) \geq {\sc F_1}(p', \examples)$.
\end{proof}

\begin{lemma}\label{lem:getnextguard}(Correctness of \ {\sc GetNextGuard}) Given inputs $\examples^+, \examples^-$, Q, \keywordss, $opt$, {\sc GetNextGuard}$(\examples^+, \examples^-, Q, \keywordss, opt)$ will eventually returns all guard programs $\psi$ with a depth limit $d_g$ such that it differentiates between $\examples^+$ and $\examples^-$, and ${\sc  UB}(\psi, \examples^+) \geq opt$.

\end{lemma}

\begin{proof}

 Observe that the worklist $\worklist$ keeps track of all section locators $v$ whose ${\sc UB}(v, \examples^+) \geq opt$. Given each $v\in  \worklist$, we return guard $\psi$ only if it successfully classifies $\examples^+$ and $\examples^-$. Also since the guard generated using the section locators has no effect on the $F_1$ of the future synthesized program, every $\psi$ returns successfully classifies $\examples^+$ and  $\examples^-$ and also achieves ${\sc UB}(v, \examples^+) \geq opt$.

We now show that {\sc GetNextGuard} returns all possible guards up to depth $d_g$ in the DSL that successfully classifies $\examples^+$ and  $\examples^-$ and also achieves ${\sc UB}(v, \examples^+) \geq opt$. Not that while {\sc GetNextGuard} enumerates guards lazily, it can still be viewed as a synthesis algorithm that returns all valid guards whose upper bound is greater $opt$. In this way, we view this algorithm as a standard enumerated procedure except we do pruning on line 8. According Theorem~\ref{thm:pruning}, the pruning is correct in the sense that we will not prune any subprogram whose extension can achieve better $F_1$ score than $opt$. And also according to Lemma~\ref{lem:ub} that {\sc UB} is a upper bound of the $F_1$ score, any program whose upper bound is smaller than $opt$, its $F_1$ score is guaranteed to be smaller than $opt$. Therefore, we will enumerate all possible guard in the DSL that satisfies the specification with a depth limit. 

\end{proof}

\begin{lemma}\label{lem:synthesizeextractors}(Correctness of \ {\sc SynthesizeExtractors}) Given inputs $\examples, Q, \keywordss, opt$,  {\sc SynthesizeExtractors}($\examples, Q, \keywordss, opt$) returns all extractors in the \dsl DSL with depth limit $d_e$ that achieves the highest $F_1$ score that is greater or equal to $opt$ on $\examples$.
\end{lemma}

\begin{proof}
We first prove that at the end of each iteration, for any extractor $e \in E_o$, $e$ achieves the highest $F_1$ score that is greater than $opt$ among all extractors enumerated so far. We prove this by inducting on the iteration  of the loop. 
\begin{itemize}
    \item Base case (iteration=1): the current optimal $F_1$ score is $opt$ and if the pulled extractor (noted as $e'$)'s $F_1$ score is greater than the current optimal  score, we update $E_o$ and $s_o$ so that it only include $e'$ and $e'$'s $F_1$ score as the new optimal score respectively. If $e'$'s $F_1$ is the same as $s_o$, the algorithm append $e'$ to $E_o$. Otherwise $E_o$ will be empty set. The statement is true by the end of the first iteration.
    \item Inductive case: suppose in iteration=n, any extractor $e' \in  E_o$ achieves the optimal $F_1$ score. Given a new extractor $e'$ pulled from the worklist with performance $F_{1_{e'}}$, if $F_{1_{e'}} > s_o$, $E_o$ is reinitialize to include only $e'$ and update the $s_o$. If $F_{1_{e'}} = s_o$, $e'$ is appended to $E_o$. Using the inductive hypothesis, $E_o$ still only includes those extractors that achieve the optimal $F_1$. In the case where $F_{1_{e'}} < s_o$, using the inductive hypothesis, the statement still holds. 
\end{itemize}

We now prove that all extractors $e$ in the \dsl DSL up to depth $d_e$ such  that $ F_1(e, \examples) = s_o$ will be included in the set $E_o$.

We first show that in the end of each iteration, given a extractor $e$ pulled from $\worklist$, any extensions of $e$  will be appended to $\worklist$ if its upper bound is greater than the current value of $s_o$ ($s_{o_{curr}}$). Note that this algorithm is a standard enumerative synthesis algorithm so without pruning it will go over all possible program in the DSL up to depth $d_e$. Since we do pruning on line 9, so this statement is correct as long as the pruning procedure is correct. According to Theorem~\ref{thm:pruning}., the pruning procedure is sound and therefore any extractor that is pruned can never be extended to achieve an $F_1$ score higher than $s_{o_{curr}}$. Therefore $\worklist$ contains all $e$ in the DSL except those that is guaranteed cannot be or become sub-program of any program that achieves $F_1$ larger than than $s_{o_{curr}}$. Furthermore, since $s_o$ can only be increased overtime, let $E_{s_o}$ be the set of extractors in the DSL up to depth $d_e$ whose upper bound is greater than $s_o$ and $E_{\worklist}$ be the set of all extractors in the DSL up to depth $d_3$ that $\worklist$ has visited by the time while loop terminates, we have $E_{s_o} \subseteq E_{\worklist}$. According to Lemma~\ref{lem:ub}, $E_o \subseteq E_{s_o}$. Note that the terminating condition of the while loop is $\worklist \equiv \varnothing$ and $\worklist$ will become empty since we limit the enumeration depth to $d_e$. Therefore, $E_o$ includes all extractors $e$ in the \dsl DSL up to depth $d_e$ such  that $ F_1(e, \examples) = s_o$. 
\end{proof}

\begin{lemma}\label{lem:synthesizebranch}(Correctness of {\sc SynthesizeBranch})
Given the input $\examples^+$, $\examples^-$, Q, \keywordss, {\sc SynthesizeBranch}($\examples^+$, $\examples^-$, Q, \keywordss) will returns all \dsl programs $P$ that achieves highest $F_1$ score for $\examples^+$ and differentiate $\examples^+$ and $\examples^-$.
\end{lemma}

\begin{proof}
We first prove that any program returned achieves maximum $F_1$ score for $\examples^+$ and differentiates $\examples^+$ and $\examples^-$. We give a proof sketch since the proof is similar to the one that proves the soundness for the {\sc SynthesizeExtractor}. Observe that in each iteration we update $opt$ and $R$ if the synthesized program achieved higher $F_1$  than the current $opt$ value and assign the program to $R$ if its $F_1$ is $opt$, and these variables are not updated otherwise. Furthermore, according to Lemma~\ref{lem:synthesizeextractors}, {\sc SynthesizeExtractors} returns the extractors that achieves the highest $F_1$ score on example $\examples^+$ and according to Lemma~\ref{lem:getnextguard}, {\sc GetNextGuard} returns the guards that differentiates between $\examples^+$ and $\examples^-$.  Therefore, it keeps track of the score and the set of programs that achieves the highest $F_1$ score on $\examples^+$ so far.  

We now prove that any branch program in the \dsl DSL whose guard is up to depth $d_g$ and extractor is up to depth $d_e$ that successfully classifies $\examples^+$ and $\examples^-$ and achieves maximum $F_1$ on $\examples^+$ is included in the set of returned program.

We first show that each iteration generates all branch programs whose $F_1$ score is $opt_i$, where $opt_i$ is the optimal $F_1$ score at iteration $i$. We can prove this statement because the following hold: (1)  Lemma~\ref{lem:getnextguard} shows that {\sc GetNextGuard} will returns guard that successfully classifies $\examples^+$ and $\examples^-$ and also achieves an upper bound greater than $opt_i$. (2) Theorem~\ref{thm:pruning}  shows the pruning procedure on line 6 is correct and therefore it will not prune out any guard who may achieves an $F_1$ score greater or equal to $opt_i$. (3) Lemma~\ref{lem:synthesizeextractors} shows that {\sc SynthesizeExtractors} will return all extractor programs up to depth $d_e$ that achieves the highest $F_1$ score on $\examples'$(which is propagated from $\examples^+$) that is greater or equal to than $opt_i$. Since in each iteration we construct all the optimal extractor programs with respect to the valid guard program, we do produce all possible branch programs in the DSL that is optimal.

Since accroding to Lemma~\ref{lem:getnextguard}, {\sc GetNextGuard} will returns all valid guards $\psi$ up to depth $d_g$ whose {\sc UB}$(\psi.v, \examples^+) \geq opt'$ for $opt' \leq opt$ and also since we keep updating the value of $opt$ and $R$ with respect to the best synthesized branch program in each iteration, when {\sc SynthesizeBranch} terminates, it will return all branch program whose guards successfully classifies $\examples^+$ and $\examples^-$ and achieves highest $F_1$ score for $\examples^+$.
\end{proof}

\begin{theorem} (Guarantee of Optimality)
Let $\examples, Q, \keywordss$ be inputs to the {\sc Synthesize} procedure and let  $p$ be a \dsl  program. Then, the set of programs returned by {\sc Synthesize}($\examples, Q, \keywordss$) includes $p$ if and only if, for any other  \dsl program $p'$, the  $F_1(p)  \geq F_1(p')$. 
\end{theorem}

\begin{proof}

We first show that any program $p$ returned achieves the highest $F_1$ score on the set of examples $\examples$. We give a brief proof sketch since the proof is similar to the one that proves the soundness for the {\sc SynthesizeExtractor}. Observe that in each iteration we update $opt$ and $R$ if the synthesized program achieved higher $F_1$  than the current $opt$ value and append the program to $R$ if its $F_1$ score is $opt$, and these variables are not updated otherwise.  

Given a partition, according to  Lemma~\ref{lem:synthesizebranch}, {\sc SynthesizeBranch} returns the optimal branch programs with respect to a set of examples in the partition. Since each branch covers a disjoint set of examples, the collection of all optimal branch programs for each set of examples in a partition forms a set of optimal top-level programs with respect to the given partition. Furthermore, since each partition consists all of the provided training examples, the variable $opt$ and R keeps track of the score and set of top-level programs that achieves the highest $F_1$ score on the training examples.

We now show that the algorithm goes through all programs that achieves the optimal $F_1$ on the given training set. Since the \dsl DSL is recursive, we need to impose depth constraints in order for the algorithm to terminates. Let's suppose that the maximum depth  of a guard is $d_g$ and the maximum depth of a extractor is $d_e$ in the {\sc Synthesize} procedure. Using Lemma~\ref{lem:synthesizebranch}, we know that for each branch, the {\sc SynthesizeBranch} procedure returns all possible branch programs in the DSL that achieves highest $F_1$ score for $\examples^+$ and differentiate $\examples^+$ and $\examples^-$. Since each branch  covers a disjoint set of examples, the collection of all synthesized branch programs for each set of examples in a partition, covers all the top-level programs in the DSL that achieves optimal $F_1$ on the given partition. 

Since the algorithm will go over all optimal top-level programs for each partition and only those that achieves the highest $F_1$ score for the  given training set are added to the return variable $R$, we prove that the set of programs returned by {\sc Synthesize}($\examples, Q, \keywordss$) includes $p$ if and only if, for any other  \dsl program $p'$, the  $F_1(p)  \geq F_1(p')$.
\end{proof}

\section{Detailed Derivation of Section~\ref{sec:self-supervision}}

In this section, we describe in more detail how we derive Eq.~\ref{eqn:derivationgoal} from Eq.~\ref{eqn:selfsupervisedobj}, including the key step Eq.~\ref{eqn:derivation}.

\paragraph{Assumptions.}

We assume the standard probabilistic model from the semi-supervised learning literature~\cite{semisuperviseassump}:
\small
\begin{align*}
p(i,o,\pi)=p(o\mid\pi,i)\cdot p(\pi)\cdot p(i),
\end{align*}
\normalsize
where $i$ is an input, $o$ is an output, and $\pi$ is a program. In addition, we assume that
\small
\begin{align*}
p(o\mid\pi,i)&=\mathbbm{1}(o=\pi'(i)) \\
p(\pi)&=|\Pi|^{-1},
\end{align*}
\normalsize
where $\Pi$ is the space of all possible programs (which is finite since we consider programs of bounded depth). In other words, we assume $p(o\mid\pi,i)$ is only non-zero when $o$ is the output of $\pi$. Next, we note that $p(i)$ is the data distribution, so we do not need to model it. In addition, we also assume that two different examples $(i,o)$ and $(i',o')$ are conditionally independent given $\pi$---i.e.,
\small
\begin{align*}
p(i,o,i',o',\pi)=p(o\mid\pi,i)\cdot p(o'\mid\pi,i')\cdot p(\pi).
\end{align*}
\normalsize
Finally, we let $\Pi^*$ denote the set of programs that are correct for all examples $(i',o')\in\examples$---i.e.,
\small
\begin{align*}
\Pi^*=\{\pi\in\Pi\mid\forall(i',o')\in\examples\;.\;o'=\pi(i')\}.
\end{align*}
\normalsize
In practice $\Pi^*$ may be empty (i.e., if there are no programs that satisfy all the given examples $(i',o')\in\examples$), so we approximate it using the set of programs that achieve optimal loss (e.g., according to the $F_1$ score). This set might be very large, so we additionally approximate it using samples $\Pi_E$. This approximation is implicitly used in Section~\ref{sec:self-supervision}.

\paragraph{Theoretical analysis.}

We show the following result:
\begin{theorem}
Letting
\small
\begin{align*}
\tilde{L}(\pi;\examples,\inputexamples)
&=\sum_{\outputexamples}p(\outputexamples\mid\inputexamples,\examples)\cdot L(\pi;\inputexamples,\outputexamples),
\end{align*}
\normalsize
then
\small
\begin{align}
\label{eqn:derivationgoal2}
\tilde{L}(\pi;\examples,\inputexamples)&=\frac{1}{\mathcal{N}}\sum_{j=1}^{\mathcal{N}}L(\pi;\inputexamples,\outputexamples_j),
\end{align}
\normalsize
where $\mathcal{N}=|\Pi^*|$, and where
\small
\begin{align*}
\outputexamples_j=(\pi_j(\inputexamplee_1),...,\pi_j(\inputexamplee_K))\qquad(\forall\pi_j\in\Pi^*).
\end{align*}
\normalsize
\end{theorem}
We note that Eq.~\ref{eqn:derivationgoal2} is identical to Eq.~\ref{eqn:derivationgoal}, except in Eq.~\ref{eqn:derivationgoal} we have taken $\Pi^*$ to be the set of programs with optimal $F_1$ score on $\examples$, and have furthermore approximated this set using samples $\Pi_E$ from $\Pi^*$.

\begin{proof}
First, by our conditional independence assumption, given program $\pi$, unlabeled input examples $\inputexamples$, candidate output labels $\outputexamples$, and labeled examples $\examples$, we have
\small
\begin{align*}
&p(\inputexamples, \outputexamples, \examples, \pi) \\
&= p(\outputexamples, \pi\mid\inputexamples) \cdot p(\inputexamples) \cdot p(\examples \mid \pi) \cdot p(\pi) \\
&=\left(\prod_{k=1}^K p(o_k \mid \pi, i_k)\cdot p(i_k)\right) \cdot \left(\prod_{h=1}^H p(o_h' \mid \pi, i_h') \cdot p(i_h')\right) \cdot p(\pi),
\end{align*}
\normalsize
where $\inputexamples=(i_1,...,i_K)$, $\outputexamples=(o_1,...,o_K)$, and $\examples=(\inputexamples',\outputexamples')$, and where $\inputexamples'=(i_1',...,i_H')$, and $\outputexamples'=(o_1',...,o_H')$. In other words, $\examples$ is conditionally independent of $(\inputexamples,\outputexamples)$ given $\pi$.

Now, we proceed with our proof. First, by the law of total probability, we have
\small
\begin{align}
\label{eqn:derivation1}
p(\outputexamples\mid\inputexamples,\examples)=\sum_{\pi'\in\Pi}p(\pi'\mid\inputexamples,\examples)\cdot p(\outputexamples\mid\pi',\inputexamples,\examples).
\end{align}
\normalsize
To simplify Eq.~\ref{eqn:derivation1}, we show that $p(\outputexamples\mid\inputexamples,\examples,\pi')=p(\outputexamples\mid\inputexamples,\pi')$, and that $p(\pi'\mid\inputexamples,\examples)=p(\pi'\mid\examples)$. First, to show the former, note that 
\small
\begin{align*}
p(\pi' \mid \inputexamples, \examples) &= \frac{p(\inputexamples, \examples \mid \pi') \cdot p(\pi')}{p(\inputexamples, \examples)}  \\
&= \frac{p(\inputexamples) \cdot p(\examples \mid \pi') \cdot p(\pi')}{p(\inputexamples) \cdot p(\examples)} \\
&= \frac{p(\examples \mid \pi') \cdot p(\pi')}{p(\examples)} \\
&= p(\pi' \mid \examples).
\end{align*}
\normalsize
Similarly, to show the latter, we have
\small
\begin{align*}
p(\outputexamples \mid \inputexamples, \examples, \pi') &= \frac{p(\outputexamples, \examples, \inputexamples, \pi')}{p(\inputexamples, \examples, \pi')} \\
&= \frac{p(\outputexamples \mid \inputexamples, \pi') \cdot p(\inputexamples) \cdot p(\examples, \pi')}{p(\inputexamples) \cdot p(\examples, \pi')} \\
&= p(\outputexamples \mid \inputexamples, \pi').
\end{align*}
\normalsize
Thus, plugging into Eq.~\ref{eqn:derivation1}, we have
\small
\begin{align*}
p(\outputexamples\mid\inputexamples,\examples)=\sum_{\pi'\in\Pi}p(\pi'\mid\examples)\cdot p(\outputexamples\mid\pi',\inputexamples).
\end{align*}
\normalsize
Note that this equation is identical to Eq.~\ref{eqn:derivation}. Next, by definition of $p(o\mid\pi,i)$, we have
\small
\begin{align*}
p(\outputexamples\mid\inputexamples,\pi')=\prod_{k=1}^K \mathbbm{1}(\outputexamplee_k=\pi'(\inputexamplee_k)),
\end{align*}
\normalsize
so it follows that
\small
\begin{align}
\label{eqn:derivation2}
p(\outputexamples \mid \inputexamples, \examples)=\sum_{\pi'\in\Pi}p(\pi'\mid\examples)\cdot\prod_{k=1}^K \mathbbm{1}(\outputexamplee_k=\pi'(\inputexamplee_k)).
\end{align}
\normalsize
It remains to compute $p(\pi'\mid\examples)$. To this end, we have
\small
\begin{align*}
p(\pi' \mid \examples) &= \frac{p(\inputexamples', \outputexamples', \pi')}{p(\examples)} \\
&= \frac{p(\outputexamples' \mid \pi', \inputexamples')\cdot p(\pi') \cdot p(\inputexamples)}{p(\examples)} \\
&= \frac{\left(\prod_{h=1}^H \mathbbm{1}(o_h'=\pi'(i_h'))\right)\cdot p(\inputexamples)}{|\Pi|\cdot p(\examples)}  \\
&= \frac{\mathbbm{1}(\pi'\in\Pi^*) \cdot p(\inputexamples)}{|\Pi^*|\cdot|\Pi|\cdot p(\examples)}.
\end{align*}
\normalsize
Thus, letting $\mathcal{N} = |\Pi| \cdot |\Pi^*| \cdot p(\examples) / p(\inputexamples)$, we have
\small
\begin{align}
\label{eqn:derivation3}
p(\pi' \mid \examples) = \frac{\mathbbm{1}(\pi' \in \Pi^*)}{\mathcal{N}}.
\end{align}
\normalsize
Note that since $\sum_{\pi'\in\Pi}p(\pi'\mid\examples)=1$, we must have $\mathcal{N}=|\Pi^*|$. Plugging Eq.~\ref{eqn:derivation3} into Eq.~\ref{eqn:derivation2}, we have
\small
\begin{align*}
p(\outputexamples\mid\inputexamples,\examples)=\frac{1}{\mathcal{N}}\sum_{\pi'\in\Pi^*}\prod_{k=1}^K\mathbbm{1}(o_k=\pi'(i_k)).
\end{align*}
\normalsize
The remaining steps follow Section~\ref{sec:self-supervision}. In particular, by the definition of $\outputexamples_j$, we have
\small
\begin{align*}
p(\outputexamples\mid\inputexamples,\examples)
&=\frac{1}{\mathcal{N}}\sum_{j=1}^{\mathcal{N}}\mathbbm{1}(\outputexamples=\outputexamples_j),
\end{align*}
from which it follows that
\begin{align*}
\tilde{L}(\pi;\examples,\inputexamples)
&=\sum_{\outputexamples}p(\outputexamples\mid\inputexamples,\examples)\cdot L(\pi;\inputexamples,\outputexamples) \nonumber  \\
&=\frac{1}{\mathcal{N}}\sum_{j=1}^{\mathcal{N}}L(\pi;\inputexamples,\outputexamples_j),
\end{align*}
\normalsize
as claimed.
\end{proof}

\section{Additional ablation studies}

To help readers better understand the design choices behind \toolname, we  present additional ablation studies evaluating the impact of  the different input modalities used by \toolname as well as its sensitivity to the number of labeled webpages. 

\begin{figure}[!t]
\begin{center}
        \definecolor{coolgray}{RGB}{38,70,83}
    \definecolor{powderblue}{RGB}{190,227,219}
    \definecolor{peachbeige}{RGB}{255,214,186}
    \definecolor{champagnepink}{RGB}{237,221,212}
\begin{tikzpicture}[scale=0.9]
\begin{axis}[
    ymin=0,
    ymax=0.8,
    ybar,
    enlarge x limits=0.20,
    legend style={at={(0.5,-0.15)},
      anchor=north,legend columns=-1},
    ylabel={Avg $F_1$},
    symbolic x coords={Faculty,Conference,Class, Clinic},
    xtick=data,
    legend image code/.code={
        \draw [#1] (0cm,-0.1cm) rectangle (0.2cm,0.25cm); },
    ]
\addplot[black,fill=champagnepink
] coordinates {(Faculty, 0.22) (Conference, 0.42) (Class, 0.35) (Clinic, 0.49)};
\addplot[black,fill=coolgray
] coordinates {(Faculty, 0.68) (Conference, 0.57) (Class, 0.52) (Clinic, 0.647)};
\addplot[black,fill=powderblue
] coordinates {(Faculty, 0.74) (Conference, 0.70) (Class, 0.69) (Clinic, 0.657)};
\legend{\toolname-{\sc NL},\toolname-{\sc KW},\toolname}
\end{axis}
\end{tikzpicture}
    \captionsetup{font={small}}
    \caption{Comparison between \toolname and its variants}
    \label{fig:ablation-input}
\end{center}
\end{figure}
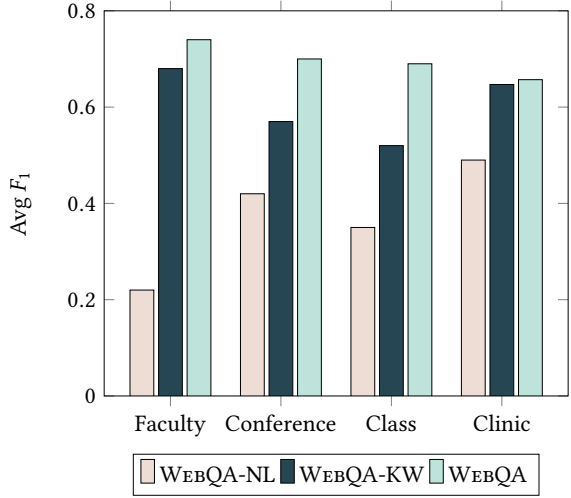

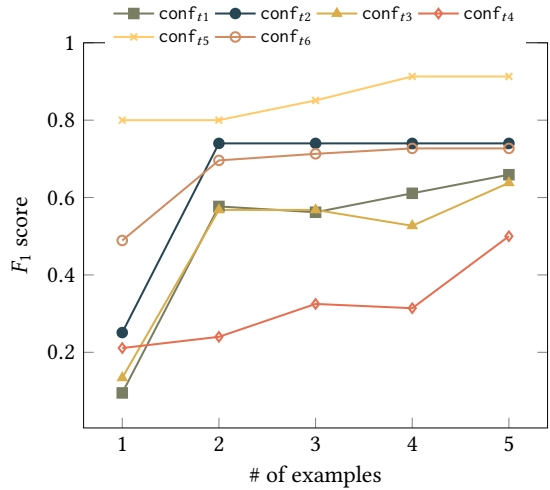
\begin{figure}[!t]

\begin{center}
            \definecolor{warmgrey}{RGB}{121,125,98}
    \definecolor{coolgray}{RGB}{38,70,83}
    \definecolor{deeporange}{RGB}{244,162,97}
    \definecolor{cadmiumorange}{RGB}{231,111,81}
    \definecolor{greyolive}{RGB}{155,155,122}
    \definecolor{peachbeige}{RGB}{217,174,72}
    \definecolor{yellowbeige}{RGB}{241,220,167}
    \definecolor{yellowochre}{RGB}{255,203,105}
    \definecolor{nectar}{RGB}{208,140,96}
    \begin{tikzpicture}[scale=0.9]
        \begin{axis}[
            ylabel near ticks,
            ymax= 1.0,
            legend style = {nodes={scale=0.8, transform shape}, anchor=north, at={(0.5, 1.12)}, legend columns=4, draw=none},
            xlabel style={yshift=1mm},
            ylabel = $F_1$ score,
            xlabel = \# of examples,
        ]
        \legend{${\tt conf}_{t1}$, ${\tt conf}_{t2}$, ${\tt conf}_{t3}$, ${\tt conf}_{t4}$, ${\tt conf}_{t5}$, ${\tt conf}_{t6}$}
        \addplot[
            line width=0.3mm,
            mark=square*,
            color=warmgrey
        ] coordinates {
            (5, 0.659)
            (4, 0.611)
            (3, 0.562)
            (2, 0.577)
            (1, 0.095)
        };
        \addplot[
            line width=0.3mm,
            mark=*,
            color=coolgray
        ] plot coordinates {
            (5, 0.740)
            (4, 0.740)
            (3, 0.740)
            (2, 0.740)
            (1, 0.251)
        };

        \addplot[
            line width=0.3mm,
            mark=triangle*,
            color=peachbeige
        ] plot coordinates {
            (5, 0.638)
            (4, 0.527)
            (3, 0.568)
            (2, 0.568)
            (1, 0.134)
        };
        
        \addplot[
            line width=0.3mm,
            mark=diamond,
            color=cadmiumorange
        ] plot coordinates {
            (5, 0.500)
            (4, 0.314)
            (3, 0.325)
            (2, 0.240)
            (1, 0.211)
        };
        \addplot[
            line width=0.3mm,
            mark=x,
            color=yellowochre
        ] plot coordinates {
            (5, 0.913)
            (4, 0.913)
            (3, 0.851)
            (2, 0.800)
            (1, 0.800)
        };
        \addplot[
            line width=0.3mm,
            mark=o,
            color=nectar
        ] plot coordinates {
            (5, 0.727)
            (4, 0.727)
            (3, 0.713)
            (2, 0.696)
            (1, 0.489)
        };
        \end{axis}
    \end{tikzpicture}
    
        \captionsetup{font={small}}
        \caption{$F_1$ score achieved in each task of the Conference domains with respect to the number of labeled examples.}
        \label{fig:ablation-ex}
\end{center}
\end{figure}

\subsection{Evaluation on the types of input}

Recall that \toolname takes two types of queries as input: a question and a set of keywords. In this section, we evaluate  the impact of these two types of inputs on the end-to-end performance of the tool. Specifically,  \figref{fig:ablation-input} shows the average $F_1$ score for each evaluation domain for the following two variants of \toolname:
\begin{itemize}
    \item \toolname-{\sc NL}: This variant only uses the question but not the keywords.
    \item \toolname-{\sc KW}: This variant only uses the keywords but not the question. 
\end{itemize}

As we can see from \figref{fig:ablation-input}, the system works the best when both modalities of inputs are utilized. We also performed 1-tailed $t$-tests to check whether the differences in performance are significant and obtained p-values less than $0.01$ in the comparison to the two variants. Thus, these results provide evidence that using a combination of questions and keywords as inputs leads to more accurate  results.

\subsection{Evaluation on the number of labeled webpages}
In this section, we evaluate \toolname's sensitivity to the number of labeled examples. For this evaluation, we focus on all $6$ tasks in the conference domain and vary the number of training examples from one to five. Specifically, we obtain these examples by  removing a subset of the labeled webpages used in our evaluation from Section~\ref{sec:eval}.

Our results are presented in \figref{fig:ablation-ex}. This graph shows the $F_1$ score  (y-axis) with respect to the number of labeled examples  (x-axis). As shown in  \figref{fig:ablation-ex},  while  performance generally gets worse as we reduce the number of examples, sensitivity to the number of examples varies from task to task. For example, for the ${\tt conf}_{t5}$ task, \toolname is able to synthesize programs that achieve high $F_1$ with only a single labeled example, whereas $F_1$ score drops significantly for  ${\tt conf}_{t4}$  if we remove even one of the examples. 


\clearpage
\section{Task-wise stats for Section 8.1 }

\begin{table*}[b]
\centering
\footnotesize
\caption{Questions and Keywords for each task.}
\begin{tabular}{|c|c|l|l|}
\hline
    Domain & Task & Question & Keywords  \\
    \hline 
    \multirow{8}{*}{Faculty} & ${\tt fac}_{t1}$ & Who are the current PhD students? & Current Students, PhD \\
    & ${\tt fac}_{t2}$ & What are the conference publications at PLDI? &  Conference Publications,   PLDI \\
    & ${\tt fac}_{t3}$ & What courses does this person teach? & Courses, Teaching \\
    & ${\tt fac}_{t4}$ & What are the the papers that received the Best Paper Award? &  Conference Publications,  Best Paper Award \\
    & ${\tt fac}_{t5}$ & What program committees or PC has this person served for? & Program Committee, PC \\
    & ${\tt fac}_{t6}$ & What conference papers have been published in 2012? &  Conference Publications,  2012 \\
    & ${\tt fac}_{t7}$ & Who are the co-authors among all papers published at PLDI? &  Conference Publications,  PLDI \\
    & ${\tt fac}_{t8}$ & Who are the alumni or formerly advised students? & Alumni, Former Students \\
    \hline
    \multirow{6}{*}{Conference} & ${\tt conf}_{t1}$ & Who are the program chairs or co-chairs? & Program Chair, Program Co-chair, PC Chair \\
    & ${\tt conf}_{t2}$ & Who are the program committee (PC) members? & Program Committee, PC \\
    & ${\tt conf}_{t3}$ & What are the topics of interest? & Topics \\
    & ${\tt conf}_{t4}$ & When is the paper submission deadline? & Paper Submission Deadline \\
    & ${\tt conf}_{t5}$ & Is this conference double-blind or single-blind? & Double-blind, Single-blind \\
    & ${\tt conf}_{t6}$ & What institutions are the program committee or PC members from? & Program Committee, PC \\
    \hline
    \multirow{6}{*}{Class} & ${\tt class}_{t1}$ & When are the lectures or sections? & Section, Lecture \\ 
    & ${\tt class}_{t2}$ & Who are the instructors? & Instructors \\ 
    & ${\tt class}_{t3}$ & Who are the teaching assistants (TAs)? & Teaching Assistants, TAs \\ 
    & ${\tt class}_{t4}$ & When are the midterms or exams? & Exam, Midterm, Test \\ 
    & ${\tt class}_{t5}$ & What are the textbooks? & Textbooks, Materials, Required Texts \\ 
    & ${\tt class}_{t6}$ & How are the grades counted in this class? & Grades, Grading, Rubric \\ 
    \hline
    \multirow{5}{*}{Clinic}  & ${\tt clinic}_{t1}$ & Who are the doctors or providers? & Doctor, Provider, Our Team \\
    & ${\tt clinic}_{t2}$ & What types of service do they provide? &  Our Services \\
    & ${\tt clinic}_{t3}$ & What types of treatments do they specialize in? & Treatments, Specialties \\
    & ${\tt clinic}_{t4}$ & What insurance plan do they accept? & Insurance, Plans Accepted \\
    & ${\tt clinic}_{t5}$ & Where are the clinics located? & Locations \\
    \hline
    
\end{tabular}
\end{table*}

\begin{table*}[b]
\footnotesize
    \centering
    \caption{Evaluation results for each baseline per task. P stands for Precision and R means Recall.}
    \begin{tabular}{|c|c|ccc|ccc|ccc|ccc|}
    \hline
    \multirow{2}{*}{Domain} & \multirow{2}{*}{Task} & \multicolumn{3}{c|}{\toolname} & \multicolumn{3}{c|}{{\sc BERTQA}} & \multicolumn{3}{c|}{{\sc HYB}} & \multicolumn{3}{c|}{{\sc EntExtract}} \\
    & & P & R & $F_1$ & P & R & $F_1$ & P & R & $F_1$ & P & R & $F_1$ \\
    \hline
    \multirow{8}{*}{Faculty} & ${\tt fac}_{t1}$  & 0.49 & 0.82 & 0.62 & 0.72 & 0.42 & 0.53 & 1.0 & 0.03 & 0.05 & 0 & 0.05 & 0 \\
    & ${\tt fac}_{t2}$ & 0.80 & 0.91 & 0.85 & 0.45 & 0.10 & 0.16 & 0 & 0 & 0 & 0.06 & 0.27 & 0.1 \\
    & ${\tt fac}_{t3}$  & 0.74 & 0.86 & 0.80 & 0.61 & 0.05 & 0.10 & 0.32 & 0.05 & 0.16 & 0.02 & 0.08 & 0.03 \\
    & ${\tt fac}_{t4}$ & 0.76 & 0.56 & 0.65 & 0.33 & 0.15 & 0.21 & 0.5 & 0.03 & 0.05 & 0.02 & 0.26 & 0.03 \\
    & ${\tt fac}_{t5}$ & 0.78 & 0.94 & 0.85 & 0.43 & 0.04 & 0.07 & 0 & 0 & 0 & 0.01 & 0.03 & 0.02 \\
    & ${\tt fac}_{t6}$ & 0.76 & 0.86 & 0.81 & 0.44 & 0.05 & 0.09 & 1.0 & 0.03 & 0.06 & 0.07 & 0.28 & 0.12 \\
    & ${\tt fac}_{t7}$ & 0.87 & 0.90 & 0.88 & 0.37 & 0.16 & 0.22 & 0 & 0 & 0.01 & 0 & 0.07 & 0.01 \\
    & ${\tt fac}_{t8}$ & 0.55 & 0.59 & 0.57 & 0.15 & 0.04 & 0.06 & 1.0 & 0.01 & 0.02 & 0 & 0.06 & 0.01 \\
    \hline
    \multirow{6}{*}{Conference} & ${\tt conf}_{t1}$ & 0.61 & 0.72 & 0.66 & 0.36 & 0.51 & 0.42 & 0 & 0 & 0 & 0 & 0.26 & 0.01 \\
    & ${\tt conf}_{t2}$ & 0.82 & 0.67 & 0.74 & 0.28 & 0.01 & 0.01 & 1.0 & 0 & 0 & 0.31 & 0.53 & 0.39 \\
    & ${\tt conf}_{t3}$ & 0.76 & 0.55 & 0.64 & 0.85 & 0.08 & 0.15 & 0.36 & 0.07 & 0.12 & 0.03 & 0.14 & 0.05 \\
    & ${\tt conf}_{t4}$ & 0.45 & 0.57 & 0.50 & 0.61 & 0.42 & 0.50 & 0.21 & 0.02 & 0.04 & 0 & 0.07 & 0 \\
    & ${\tt conf}_{t5}$ & 0.88 & 0.95 & 0.91 & 0.86 & 0.82 & 0.84 & 0 & 0 & 0 & 0 & 0 & 0 \\
    & ${\tt conf}_{t6}$ & 0.78 & 0.68 & 0.72 & 0.50 & 0.01 & 0.02 & 0 & 0 & 0 & 0.08 & 0.21 & 0.12 \\
    \hline
    \multirow{6}{*}{Class} & ${\tt class}_{t1}$ & 0.47 & 0.81 & 0.59 & 0.63 & 0.4 & 0.52 & 0.13 & 0.13 & 0.13 & 0 & 0.02 & 0 \\
    & ${\tt class}_{t2}$ & 0.74 & 0.71 & 0.72 & 0.57 & 0.61 & 0.59 & 0.15 & 0.04 & 0.06 & 0 & 0.05 & 0.01 \\
    & ${\tt class}_{t3}$ & 0.89 & 0.80 & 0.85 & 0.67 & 0.16 & 0.26 & 0 & 0 & 0 & 0.12 & 0.32 & 0.17 \\
    & ${\tt class}_{t4}$ & 0.62 & 0.74 & 0.67 & 0.28 & 0.21 & 0.24 & 0 & 0 & 0 & 0.01 & 0.10 & 0.02 \\
    & ${\tt class}_{t5}$ & 0.64 & 0.62 & 0.63 & 0.54 & 0.11 & 0.18 & 0.05 & 0.06 & 0.05 & 0.01 & 0.02 & 0.02 \\
    & ${\tt class}_{t6}$ & 0.44 & 0.97 & 0.61 & 0.61 & 0.05 & 0.1 & 0.75 & 0.01 & 0.01 & 0.09 & 0.06 & 0.07 \\
    \hline
    \multirow{5}{*}{Clinic} & ${\tt clinic}_{t1}$ & 0.80 & 0.83 & 0.82 & 0 & 0 & 0 & 0 & 0 & 0 & 0.32 & 0.55 & 0.40 \\
    & ${\tt clinic}_{t2}$ & 0.58 & 0.55 & 0.56 & 0.32 & 0.03 & 0.05 & 0.39 & 0.19 & 0.26 & 0.05 & 0.11 & 0.07 \\
    & ${\tt clinic}_{t3}$ & 0.62 & 0.47 & 0.53 & 0.19 & 0.02 & 0.04 & 0.72 & 0.06 & 0.12 & 0.02 & 0.07 & 0.03 \\
    & ${\tt clinic}_{t4}$ & 0.83 & 0.53 & 0.65 & 0.61 & 0.01 & 0.02 & 0 & 0 & 0 & 0.28 & 0.23 & 0.26 \\
    & ${\tt clinic}_{t5}$ & 0.69 & 0.77 & 0.73 & 0.44 & 0.05 & 0.08 & 1.0 & 0.03 & 0.06 & 0.01 & 0.05 & 0.02 \\
    \hline
    \end{tabular}
    \label{tab:domain-results}
\end{table*}

\end{document}